\documentclass[letterpaper,11pt]{article}

\usepackage{diagbox}
\usepackage{makecell}
\usepackage{booktabs}
\usepackage{tikz,tikz-3dplot}
\usepackage{amsmath}
\usepackage{bbm}
\usepackage{amsfonts}
\usepackage{amssymb}
\usepackage{amsthm}
\usepackage{url,ifthen}
\usepackage[shortlabels]{enumitem}
\usepackage{srcltx}
\usepackage{dsfont}
\usepackage{multirow}
\usepackage{boxedminipage}
\usepackage[margin=1.1in]{geometry}
\usepackage{nicefrac}
\usepackage{xspace}
\usepackage{graphicx}
\usepackage{color}
\usepackage{subfigure}
\usepackage{colortbl}
\usepackage{setspace}
\usepackage{pgfplots}
\pgfplotsset{compat=1.12}
\usepackage{algorithm,algorithmic}
\usepackage[algo2e]{algorithm2e}
\RestyleAlgo{ruled}

\usepackage{xcolor}
\definecolor{DarkGreen}{rgb}{0.1,0.5,0.1}
\definecolor{DarkRed}{rgb}{0.5,0.1,0.1}
\definecolor{DarkBlue}{rgb}{0.1,0.1,0.5}
\definecolor{Gray}{rgb}{0.2,0.2,0.2}

\definecolor{c1}{RGB}{38, 70, 83}
\definecolor{c2}{RGB}{42, 157, 143}
\definecolor{c3}{RGB}{233, 196, 106}
\definecolor{c5}{RGB}{231, 111, 81}
\definecolor{c4}{RGB}{244, 162, 97}

\definecolor{c1}{RGB}{38, 70, 83}
\definecolor{c2}{RGB}{42, 157, 143}
\definecolor{c3}{RGB}{233, 196, 106}
\definecolor{c5}{RGB}{231, 111, 81}
\definecolor{c4}{RGB}{244, 162, 97}

\usepackage{listings}
\lstdefinestyle{mystyle}{
    commentstyle=\color{DarkBlue},
    keywordstyle=\color{DarkRed},
    numberstyle=\tiny\color{Gray},
    stringstyle=\color{DarkGreen},
    basicstyle=\footnotesize,
    breakatwhitespace=false,         
    breaklines=true,                 
    captionpos=b,                    
    keepspaces=true,                 
    numbers=left,                    
    numbersep=5pt,                  
    showspaces=false,                
    showstringspaces=false,
    showtabs=false,                  
    tabsize=2
}
\usepackage[small]{caption}
\usepackage[pdftex]{hyperref}
\hypersetup{
    unicode=false,          
    pdftoolbar=true,        
    pdfmenubar=true,        
    pdffitwindow=false,      
    pdfnewwindow=true,      
    colorlinks=true,       
    linkcolor=DarkBlue,          
    citecolor=DarkGreen,        
    filecolor=DarkRed,      
    urlcolor=DarkBlue,          
    pdftitle={},
    pdfauthor={},
}

\def\draft{1}

\def\submit{0}

\ifnum\draft=1 
    \def\ShowAuthNotes{1}
\else
    \def\ShowAuthNotes{0}
\fi

\ifnum\submit=1
\newcommand{\forsubmit}[1]{#1}
\newcommand{\forreals}[1]{}
\else
\newcommand{\forreals}[1]{#1}
\newcommand{\forsubmit}[1]{}
\fi

\ifnum\ShowAuthNotes=1
\newcommand{\authnote}[2]{{ \footnotesize \bf{\color{DarkRed}[#1's Note:
{\color{DarkBlue}#2}]}}}
\else
\newcommand{\authnote}[2]{}
\fi

\newtheorem{theorem}{Theorem}[section]

\newtheorem{corollary}[theorem]{Corollary}

\newtheorem{fact}[theorem]{Fact}

\newtheorem{definition}[theorem]{Definition}

\newtheorem*{definition*}{Definition}
\newtheorem*{proposition*}{Proposition}

\theoremstyle{definition}
\newtheorem*{example*}{Example}
\newtheorem{example}[theorem]{Example}

\newtheoremstyle{example_contd}
{\topsep} {\topsep}%
{}
{}
{\bfseries}
{.}
{1em}
{\thmname{#1} \thmnumber{ #2}\thmnote{#3} (continued)}

\theoremstyle{example_contd}

\newcommand{\chapterref}[1]{\hyperref[ch:#1]{Chapter~\ref{ch:#1}}}

\newcommand{\claimref}[1]{\hyperref[claim:#1]{Claim~\ref{claim:#1}}}

\newcommand{\corollaryref}[1]{\hyperref[cor:#1]{Corollary~\ref{cor:#1}}}

\newcommand{\definitionref}[1]{\hyperref[def:#1]{Definition~\ref{def:#1}}}

\newcommand{\equationref}[1]{\hyperref[eq:#1]{Equation~\ref{eq:#1}}}

\newcommand{\factref}[1]{\hyperref[fact:#1]{Fact~\ref{fact:#1}}}

\newcommand{\figureref}[1]{\hyperref[fig:#1]{Figure~\ref{fig:#1}}}

\newcommand{\tableref}[1]{\hyperref[tab:#1]{Table~\ref{tab:#1}}}

\newcommand{\itemref}[1]{\hyperref[item:#1]{Item~(\ref{item:#1})}}

\newcommand{\lemmaref}[1]{\hyperref[lem:#1]{Lemma~\ref{lem:#1}}}

\newcommand{\propref}[1]{\hyperref[prop:#1]{Proposition~\ref{prop:#1}}}

\newcommand{\expref}[1]{\hyperref[exp:#1]{Example~\ref{exp:#1}}}

\newcommand{\propositionref}[1]{\hyperref[prop:#1]{Proposition~\ref{prop:#1}}}

\newcommand{\remarkref}[1]{\hyperref[rem:#1]{Remark~\ref{rem:#1}}}

\newcommand{\sectionref}[1]{\hyperref[sec:#1]{Section~\ref{sec:#1}}}

\newcommand{\theoremref}[1]{\hyperref[thm:#1]{Theorem~\ref{thm:#1}}}

\newcommand{\assumptionref}[1]{\hyperref[ass:#1]{Assumption~\ref{ass:#1}}}

\usepackage[utf8]{inputenc}
\usepackage[T1]{fontenc}
\usepackage{kpfonts}
\usepackage{microtype}

\renewcommand{\Pr}{\mathrm{Pr}}

\renewcommand{\leq}{\leqslant}

\renewcommand{\geq}{\geqslant}

\newcommand{\from}{\colon}

\usepackage{bm}

\newcommand{\ignore}[1]{}

\DeclareMathOperator*{\argmax}{arg\,max}

\renewcommand{\epsilon}{\varepsilon}

\newcommand{\remove}[1]{}

\newcommand{\ind}{\mathbf{1}}

\newcommand{\1}{\ind}

\newcommand{\Exi}[2]{\mathbb{E}_{#1}\left[#2\right]}

\usepackage{fontawesome}
\usepackage{xcolor}
\usepackage{natbib}
\usepackage{tikz}
\usetikzlibrary{arrows.meta, positioning, calc, backgrounds, fit}

\definecolor{mdorange}{HTML}{B0501A}%
\definecolor{mdorangefill}{HTML}{FBF1E8}%
\definecolor{plblue}{HTML}{2B5C8A}%
\definecolor{plbluefill}{HTML}{E0E9F3}%
\definecolor{inkblack}{HTML}{1F1F1F}%
\definecolor{surfacegrey}{HTML}{444444}%

\usepackage{xcolor}
\definecolor{focusred}{RGB}{165, 28, 48}
\definecolor{learnblue}{RGB}{30, 80, 160}
\usepackage{url}
\usepackage{xspace}
\usepackage{comment}
\usepackage{subcaption} 
\usepackage{xfrac}
\usepackage{graphicx}
\usepackage{booktabs}
\usepackage{tabularx}

\usepackage[affil-it]{authblk}

\title{On the Meta-Design of Allocation Problems}

\author[1]{Unai Fischer-Abaigar\thanks{Corresponding author: \texttt{uf2034@nyu.edu}}}
\author[2]{Emily Aiken}
\author[3, 4]{Christoph Kern}
\author[1]{Juan Carlos Perdomo}

\affil[1]{New York University}
\affil[2]{University of California San Diego}
\affil[3]{LMU Munich}
\affil[4]{Munich Center for Machine Learning}

\date{}
\begin{document}

\maketitle

\begin{abstract}
There is an extensive literature that studies how to find optimal policies in resource allocation problems, taking the underlying design parameters that define the allocation, such as what data is collected, how many people can be served, and quality of service as fixed constraints. Yet, from a planner’s perspective, these design parameters are themselves optimization variables that are just as important in determining overall welfare as selecting the optimal targeting rule for a given set of constraints. This realization motivates a rich set of meta-design questions exploring how planners should make principled decisions about investments in prediction, capacity constraints, and treatment quality, all of which lie upstream of classical policy optimization. Building on initial theoretical work in this space, our paper has three main contributions. First, we formally define the broad meta-design space of resource allocation problems. Second, we develop empirical tools that enable practitioners to reliably navigate it. Third, we demonstrate the framework in two real-world case studies on German employment services and targeted cash transfer programs in Ethiopia. 
\end{abstract}
\pagenumbering{gobble}
\pagenumbering{arabic}

\section{Introduction}
Public institutions increasingly use predictive algorithms to guide the allocation of scarce societal resources. For instance, algorithms are used to target tutoring services to students at risk of dropout \citep{perdomo2025wis}, allocate public housing units \citep{vayanos2026robust}, and to  prioritize vulnerable communities for early access to new vaccines \citep{barda2020developing}. Given a specification for an allocation problem --- that is, a set of data measurements (features, outcomes), utility function, and limit on the number of individuals intervened on --- there is a large body of work that studies how to learn optimal decision rules that specify who in the population should receive resources. These studies range from learning optimal treatment assignment rules through empirical welfare maximization \citep{athey2021, kitagaw2018}, to training constrained predictors that satisfy operational or fairness requirements \citep{cotter2019}, to establishing calibration conditions under which post-processing of predictions leads to optimal downstream decisions \citep{hu2023omnipredictors}.

Before deciding on a targeting rule, a social planner inevitably needs to first decide on the parameters that define the allocation (see Figure~\ref{fig:meta-design}).   
What data is collected, how many people can be served, and the quality of interventions are ultimately design choices of a social planner. Furthermore, these \emph{meta-design} decisions are just as if not more important than the downstream choice of decision rule. 
For instance, a planner who decides to expand program capacity while keeping the targeting rule fixed may do far more good than one who perfects the targeting rule for a program that is under-allocating. 

\paragraph{Our Contributions.} Given this background, our work has three main contributions. First, we formally define the space of questions in the meta-design of data-driven allocation problems. Second, we provide tools that enable practitioners to reliably understand which components of their allocation system matter most for downstream welfare, and how to prioritize investments across them. And lastly, we illustrate our methodology in the context of two real-world case studies on German employment services and poverty targeting in Ethiopia. 

On a technical level, our analysis centers on an object we term the \emph{welfare surface}. Intuitively, the welfare surface describes how the value achieved by the optimal targeting rule for a given choice of design parameters varies with these meta-design parameters (capacity constraints, data distributions, treatment effects). Drawing on connections to classical microeconomic theory and nonlinear optimization, we show how the welfare surface empowers planners to address a range of meta-design questions. These include which design axes matter most in their current configuration, testing whether existing design choices are optimal given fixed budgets, and how to allocate budgets among competing investments when expanding a system.

As a final note, a key insight of our work is that meta-design conclusions are inherently context-dependent. As shown via our case studies, which investments matter most depends on a planner's current configuration, data, and constraints in ways that no general theory can fully anticipate. Rather than prescribing one-size-fits-all answers, we give planners the tools to derive rigorous conclusions that are valid in their own context and the framework to know when those conclusions change. Importantly, we demonstrate this on rich real-world data drawn directly from the institutional settings we study, from household poverty surveys used to target cash transfers to real-world governmental labor market records used to profile job seekers.

\begin{figure}[t]
    \centering
    \resizebox{0.75\linewidth}{!}{%
\begin{tikzpicture}[
    x=0.013177cm, y=-0.013177cm,
    >={Stealth[length=4.5pt, width=4pt]},
    figtext/.style={font=\fontsize{8}{10}\selectfont},
    figquestion/.style={font=\fontsize{8}{10}\selectfont\itshape},
    figtitle/.style={font=\fontsize{9}{11}\selectfont\bfseries},
    chip/.style={draw=mdorange, line width=0.55pt, fill=white,
        rounded corners=0.29cm, minimum width=3.0cm, minimum height=0.6cm,
        inner sep=1pt, figtext},
    card/.style={line width=0.8pt, fill=white, rounded corners=0.13cm,
        minimum height=1.11cm, inner sep=2pt, align=center,
        figtext, execute at begin node=\setlength{\baselineskip}{11pt}},
    conn/.style={draw=mdorange, line width=0.75pt},
  ]
  \fill[mdorangefill] (464,16) [rounded corners=0.18cm] -- (16,16) -- (16,316)
      [sharp corners] -- (464,316) -- cycle;
  \fill[plbluefill] (464,16) [rounded corners=0.18cm] -- (1044,16) -- (1044,316)
      [sharp corners] -- (464,316) -- cycle;

  \node[anchor=north west, inner sep=0, text=mdorange, figtitle]
      at (36,28) {Meta-Design};
  \node[anchor=north west, inner sep=0, text=inkblack, figquestion]
      at (36,62) {Which allocation problem to set up?};
  \node[anchor=north west, inner sep=0, text=plblue, figtitle]
      at (488,28) {Policy Learning};
  \node[anchor=north west, inner sep=0, text=inkblack, figquestion]
      at (488,62) {How to solve the allocation problem?};

  \node[chip] (cD) at (150,123) {$\mathcal{D}$\enspace data collected};
  \node[chip] (cu) at (150,179) {$u$\enspace treatment / utility};
  \node[chip] (ca) at (150,235) {$\alpha$\enspace program capacity};

  \node[card, draw=mdorange, minimum width=2.95cm] (prob) at (464,179)
      {\textbf{Allocation problem}\\$(\mathcal{D}, u, \alpha)$};
  \node[card, draw=plblue, minimum width=2.32cm] (rule) at (720,179)
      {\textbf{Targeting rule}\\$\pi_*$};
  \node[card, draw=inkblack, minimum width=2.00cm] (welf) at (940,179)
      {\textbf{Welfare}\\$V_*(\mathcal{D}, u, \alpha)$};

  \foreach \c in {cD, cu, ca} \draw[conn] (\c.east) -- (\c.east -| 296,0);
  \draw[conn] (296,123) -- (296,235);
  \draw[conn, ->] (296,179) -- (prob.west);

  \draw[plblue, line width=0.75pt, ->] (prob.east) -- (rule.west);
  \draw[plblue, line width=0.75pt, ->] (rule.east) -- (welf.west);

  \draw[surfacegrey, line width=0.75pt, dash pattern=on 2.5pt off 2.5pt, ->]
      (prob.south) -- (464,290) -- (940,290) -- (welf.south);
  \node[anchor=south, inner sep=0, text=inkblack, figtext]
      at (702,281) {welfare surface $(\mathcal{D}, u, \alpha) \mapsto V_*(\mathcal{D}, u, \alpha)$};
\end{tikzpicture}
    }
    \caption{A planner shapes the allocation problem through investment decisions, which feed into a policy learning algorithm that produces a targeting rule and determines the resulting welfare.}
    \label{fig:meta-design}
\vspace{-15pt}
\end{figure}

\section{Related Work}
\paragraph{Prediction-Access Ratio.} The most related prior investigations are \citet{pmlr-v235-perdomo24a} and \citet{pmlr-v267-fischer-abaigar25a}. They introduce the prediction-access ratio (PAR) to study the relative welfare gains achieved by improving prediction and expanding program capacity in resource allocation contexts. Their main results identify conditions under which the relative value of prediction is low in stylized theoretical models. We build on this work in two directions. First, we explore a broader space of meta-design questions, not just the relative welfare gains between prediction and access. Second, their analysis primarily relies on stylized Gaussian models, which allow for sharp theoretical insights but constrain external validity. We show how a planner can explore the design space of their own program, without relying on whether its primitives happen to match a prior theoretical model.

\paragraph{Design Space of Allocations.} A wider literature points to the size of this design space by examining different components of allocation systems. \citet{perdomo2025wis} find that early warning systems in Wisconsin schools can deliver much of their value without individual-level risk prediction, while \citet{Mashiat_Fowler_Das_2025} find that individually targeted eviction-prevention outreach substantially outperforms neighborhood-based targeting. On the theoretical side, \citet{shirali2024} show that individual targeting only outperforms aggregate targeting in high-budget regimes when inequality is low, and \citet{casacuberta2026goodallocationsbadestimates} find that learning a good allocation can require less data than estimating heterogeneous effects. Other contributions examine alternative notions of welfare, including frameworks that motivate randomization in allocation \citep{jain2024scarceresourceallocationsrely}.

\paragraph{Solving Allocation Problems.} A related literature focuses on \emph{solving} allocation problems once formulated. Common approaches include learning risk scores to rank individuals \citep{kleinberg2015, liu2026bridgingpredictioninterventionproblems, wang2024}, estimating counterfactual risk or treatment effects to account for causal structure \citep{coston2020counterfactualriskassessmentsevaluation, Wager03072018, künzel2019}, and learning optimal treatment rules directly \citep{kitagaw2018, athey2021, manski2004, kallus2018}. Decision-focused learning trains predictors to account for downstream optimization \citep{elmachtoub2022, ren2024decisionfocused}. Recent work clarifies how these perspectives relate, connecting estimation targets to utility under different problem formulations \citep{sharma2025comparing, liu2026bridgingpredictioninterventionproblems, FISCHERABAIGAR2024101976, kern2025reliable}. 

Our work is not primarily concerned with solving a given allocation problem, but rather with the meta-design of the allocation problem itself and the tradeoffs that arise from different design choices.

\begin{figure}[ht]
\centering
\subfigure[Local Rate of Improvement\label{fig:welfare-surface-a}]{\includegraphics[width=0.24\textwidth]{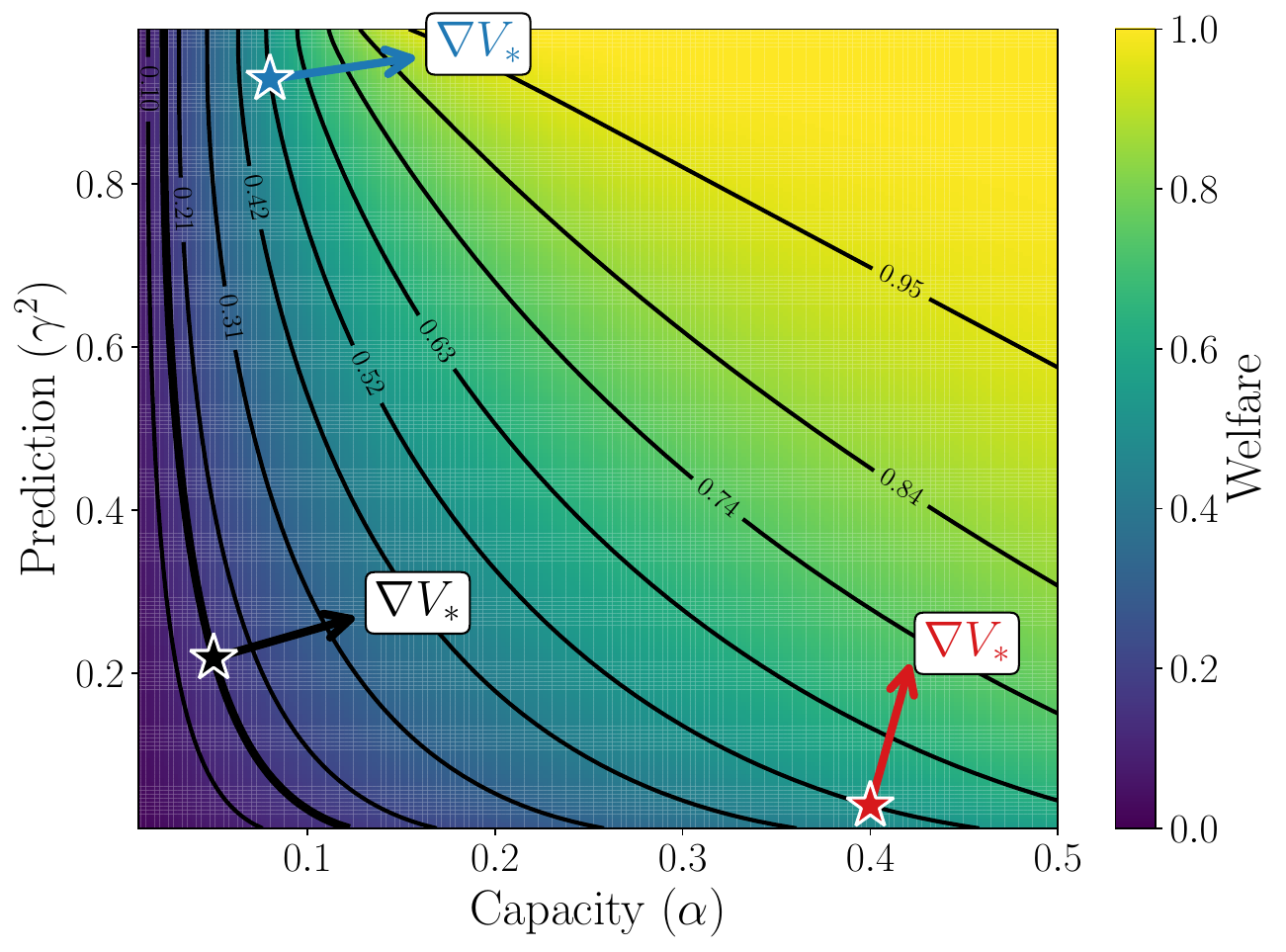}}
\subfigure[Inspecting Optimality\label{fig:welfare-surface-b}]{\includegraphics[width=0.24\textwidth]{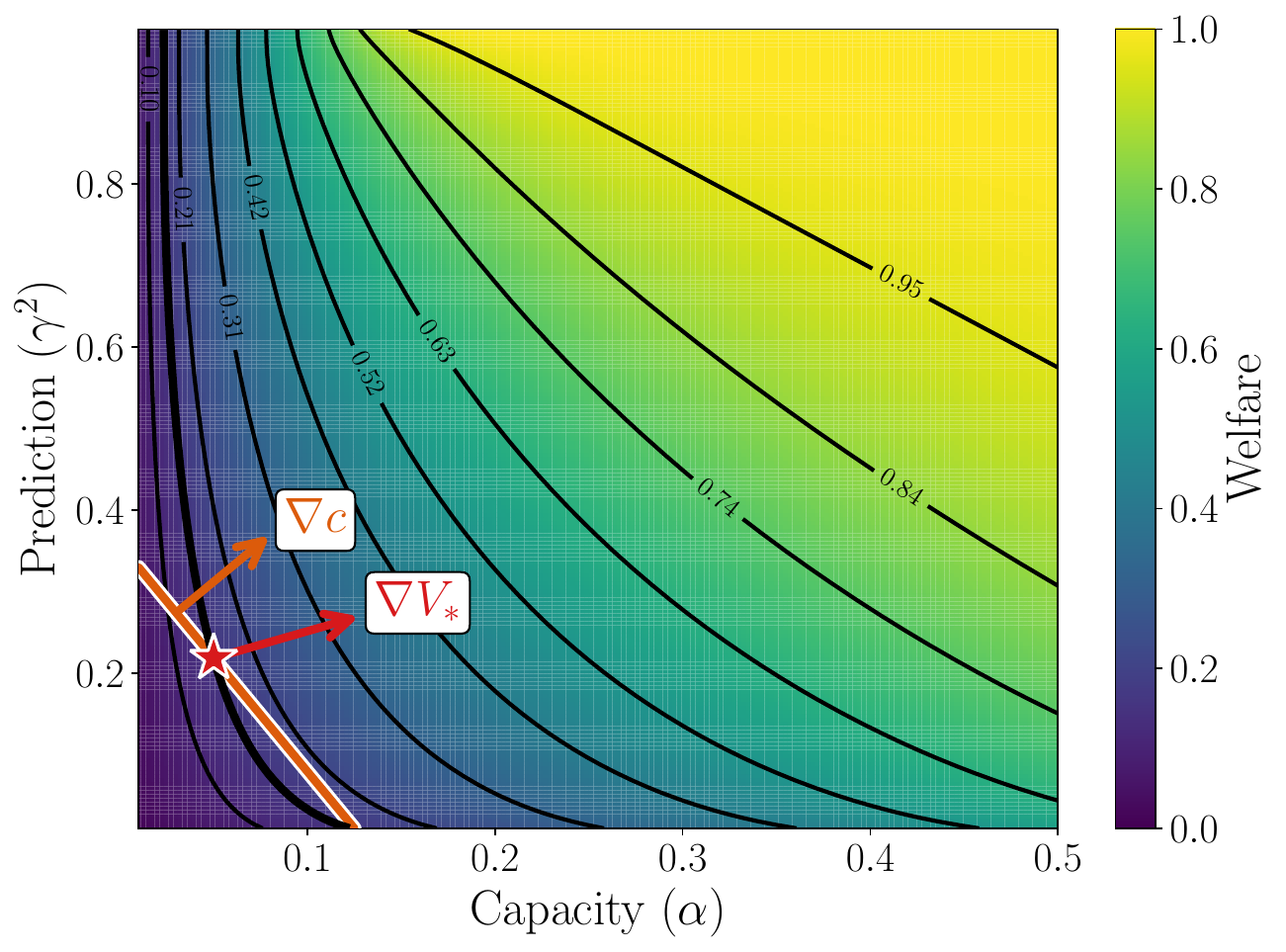}} \subfigure[Budget Optimization\label{fig:welfare-surface-c}]{\includegraphics[width=0.24\textwidth]{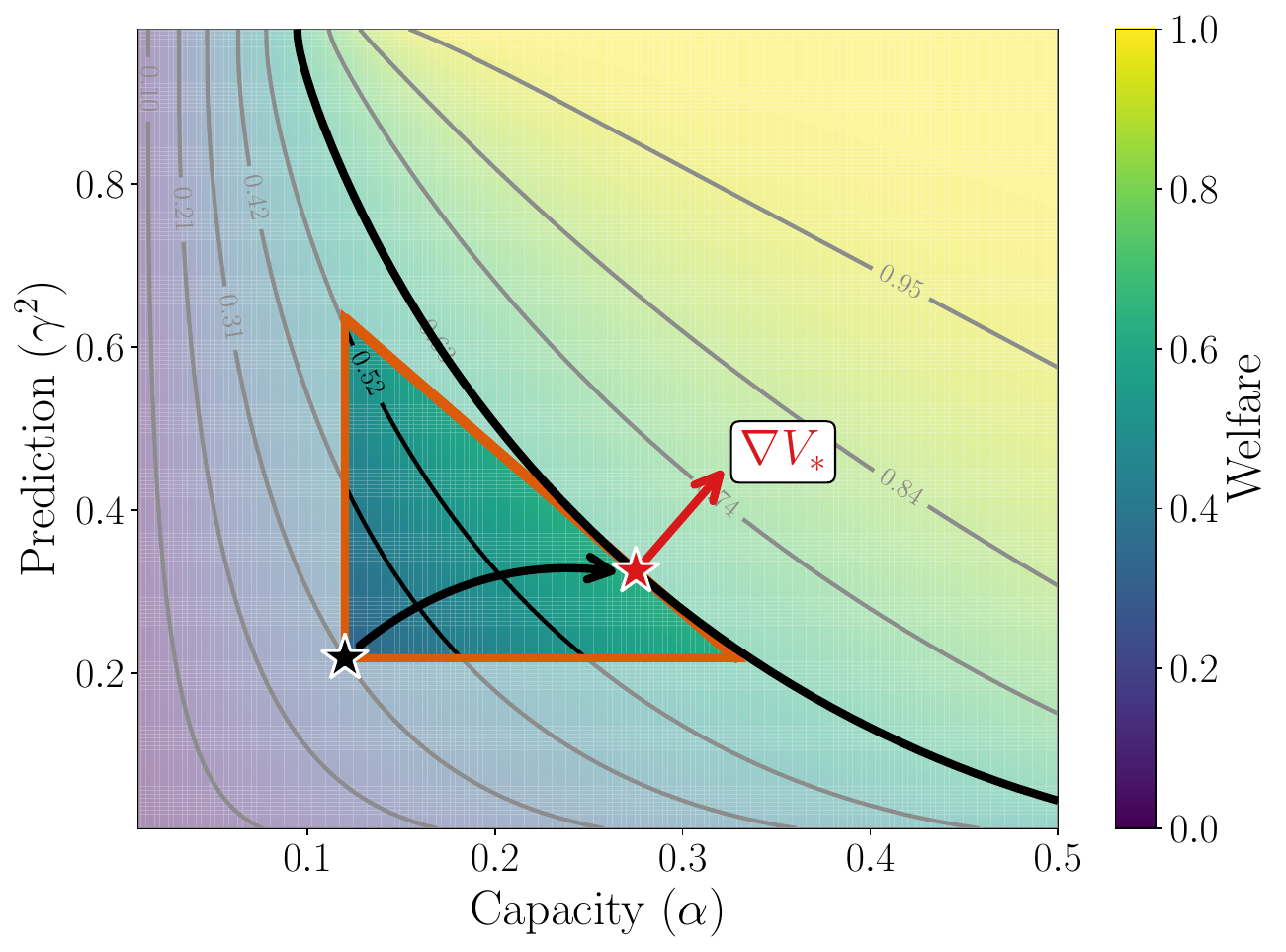}}
\subfigure[Expansion Path\label{fig:welfare-surface-d}]{\includegraphics[width=0.24\textwidth]{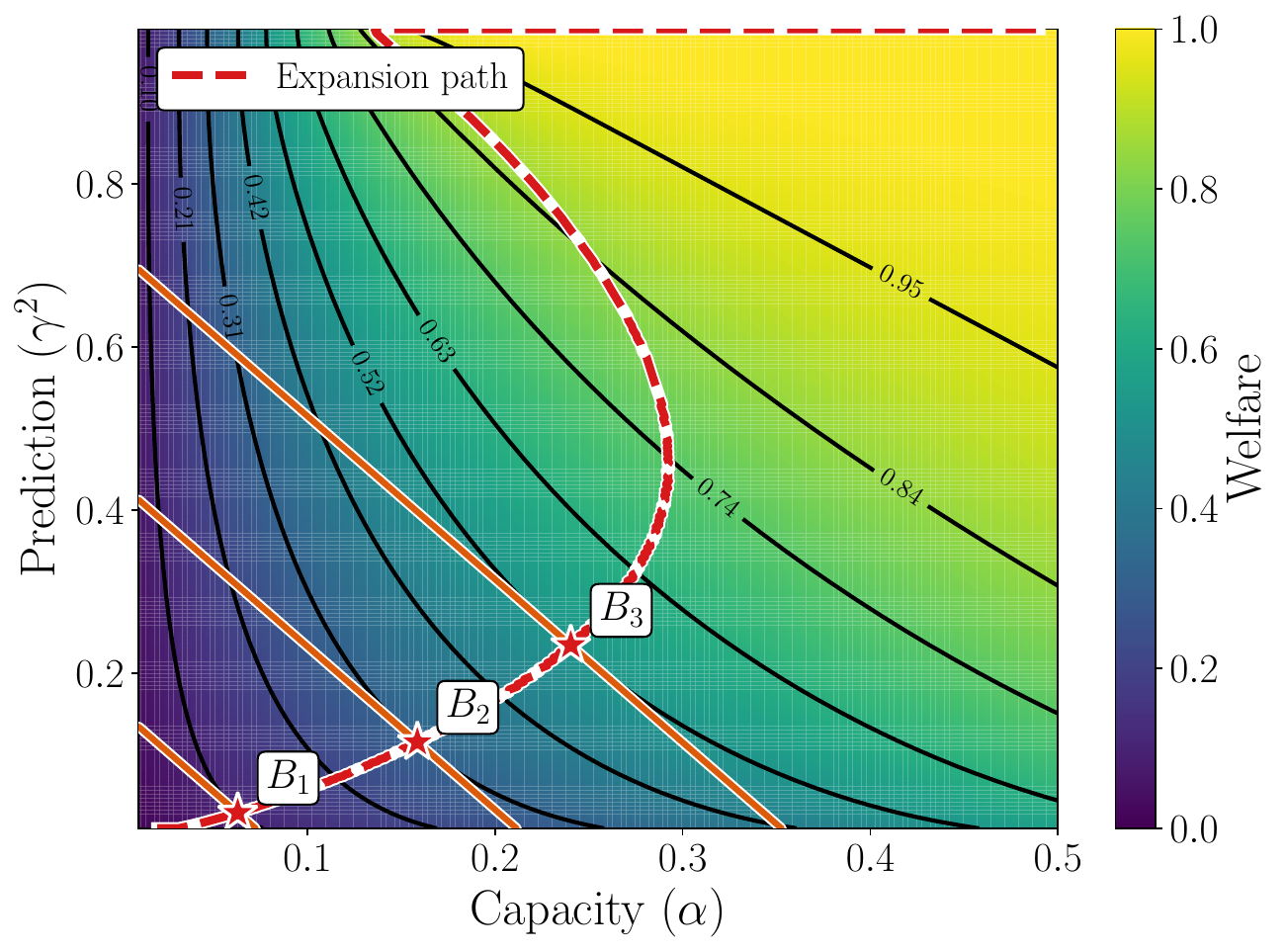}}
\caption{Welfare surface for the stylized Gaussian model of \cite{pmlr-v267-fischer-abaigar25a}, over program capacity $\alpha$ and prediction quality ($\gamma^2 $ coefficient of determination). Black lines mark iso-welfare contours. \textbf{(a)} Normalized welfare gradients at three configurations. \textbf{(b)} When the welfare gradient is not aligned with the cost gradient (orange line, constant cost), the configuration is not locally optimal. \textbf{(c)} Feasible region (orange) under additive linear costs from an existing system (black star). \textbf{(d)} The red line traces the optimal configuration as the budget grows.}
\label{fig:welfare-surface}
\end{figure}

\section{The Meta-Design of Resource Allocation Problems}
\label{sec: prediction in allocation}

Here, we formalize the meta-design space of resource allocation problems, introduce the welfare surface as the central analytical object, and outline the questions a planner can answer by studying it. 
We start our presentation by defining the optimal social welfare achieved under a fixed set of design parameters. Throughout our work, we focus on allocation problems where a social planner assigns a single scarce resource to individuals in a population, as defined below:
\begin{definition}[Allocation Problem]
\label{def:alloc}
Let $\mathcal{D}$ be a distribution over pairs $(x,y) \in \mathcal{X} \times \mathcal{Y}$ and let $\alpha \in (0,1)$ denote a capacity constraint. Fix a utility function $u: \mathcal{Y} \times \{0,1\} \rightarrow \mathbb{R}$. We define $V_*(\mathcal{D}, u, \alpha)$ to be the value of the optimal policy $\pi \from \mathcal{X} \to [0, 1]$ with respect to $(\mathcal{D}, u, \alpha)$,
\begin{align*}
    V_*(\mathcal{D}, u, \alpha) := \underset{\pi \from \mathcal{X} \to [0, 1]}{\max} \; \Exi{\mathcal{D}}{u(y, \pi(x))} \quad \text{s.t.} \; \Exi{\mathcal{D}}{\pi(x)} \leq \alpha
\end{align*}
\end{definition}
The allocation problem is defined based on three primitives; the utility function $u$, a distribution $\mathcal{D}$ which implicitly defines the set of features and outcomes, and the capacity constraint $\alpha$. The outcome $y$ and the utility function $u$ together encode the social planner's objective. Given these three primitives, $V_*(\mathcal{D}, u, \alpha)$ denotes welfare under the optimal targeting rule.

\begin{example}[Targeting the Worst-Off]
\label{worst-off-example}
Consider the problem of allocating attention to individuals deemed to be ``worst-off'' according to outcomes $y$, e.g., job seekers at risk of long-term unemployment. There, $\mathcal{D}$ is a distribution over $\mathcal{X} \times \mathbb{R}$ and the utility function $u(y, a)$ is equal to $a \cdot \1\{y \geq t_\beta\}$ where $t_\beta$ is a pre-specified cutoff and $a \in \{0,1\}$. Learning plays a role in these settings since the optimal targeting rule is to first approximate  $f_*(x) = \Pr[y\geq t_\beta|x]$, and then allocate to individuals with the highest values of $f_*(x)$. That is, $\pi_*(x) = \1\{f_*(x) \geq t_\alpha\}$ where $t_\alpha$ is chosen to satisfy the capacity constraint in Definition \ref{def:alloc}.
\end{example}

In practice, $\pi_*$ is unknown and must be learned from finite-samples. As we discussed, there is by now a rich literature \citep{bhattacharya2012inferring,slivkins2023contextual,hu2023omnipredictors} that given as input $(\mathcal{D}, u, \alpha)$ proposes algorithms for learning $\pi_*$ using i.i.d samples drawn from $\mathcal{D}$. In our work, we make use of these algorithms as a subroutine and focus on the task of optimizing over the choice of parameters $(\mathcal{D}, u, \alpha)$ themselves.

\paragraph{The Meta-Design Universe.} When building a resource allocation system, the first order of business for a planner is to decide on the resources to provide to selected individuals (for instance, the dollar amount in cash transfers), the goals of the system (jointly encoded in a utility function $u$), the information available for targeting (a data distribution $\mathcal{D}$), and the scale of the program (capacity constraints $\alpha$). In other words, planners start by selecting a precise choice of design parameters $(\mathcal{D}, u, \alpha)$ from a set, or \emph{design universe}, $\mathcal{U}$ of parameters. 

The structure of the set $\mathcal{U}$ depends on the planner's situation. When improving an existing system, $\mathcal{U}$ is naturally defined as the set of configurations reachable from a current baseline $(\mathcal{D}_0, \alpha_0, u_0)$ via available \emph{policy levers}, each lever defining an axis of improvement along which the planner can invest more or less. For example, expanding capacity by $\Delta_\alpha \in [0, \bar\alpha]$ traces out a path $\mathcal{U}_{\mathrm{cap}} = \{(\mathcal{D}_0, u_0, \alpha_0 + \Delta_\alpha) :\Delta_\alpha \in [0,\bar{\alpha}]\}$. Potential improvements in prediction define a set $\mathcal{U}_{\mathrm{pred}} = \{(\mathcal{D}_\lambda, u_0, \alpha_0,) : \lambda \in \Lambda \}$, where $\mathcal{D}_\lambda$ is a family of distributions. For example, $\lambda$ could index distributions over different feature sets. When considering improvements along multiple axis, the combined design universe is the joint parameter space, for example $\mathcal{U} = \{(\mathcal{D}_\lambda, \alpha_0, u_0) : \lambda \in \Lambda, \ \Delta_\alpha \in [0, \bar{\alpha}] \}$.

Whether building a system from scratch or improving an existing system, the design universe $\mathcal{U}$ will typically be subject to a total budget $B$, $\mathcal{U}_B = \{(\mathcal{D}, u, \alpha) \in \mathcal{U} : c(\mathcal{D}, u, \alpha) \leq B\}$ where $c(\mathcal{D}, u, \alpha)$ is the implementation cost of a given configuration. The planner's meta-design problem is then to select an element of $\mathcal{U}_B$ that maximizes downstream welfare $V_*(\mathcal{D}, u, \alpha)$.

\paragraph{The Welfare Surface.} Each element of the design universe $\mathcal{U}$ corresponds to a distinct allocation problem, and therefore a distinct level of downstream welfare under the optimal targeting rule. 
This mapping $(\mathcal{D}, \alpha, u) \rightarrow V_*(\mathcal{D}, \alpha, u)$ for $(\mathcal{D}, u, \alpha) \in \mathcal{U}$ defines the \emph{welfare surface} describing the performance of the allocation system under different configurations. 
Intuitively, the welfare surface is a landscape over the space of design parameters: high regions correspond to configurations that achieve strong downstream welfare. And as we will illustrate next, the geometry of the surface, such as its gradients and contour lines, provides insight into the relative value of different policy levers.

Given a dataset $S =\{(x_i,y_i)\}_{i=1}^n$ the welfare surface can be \textit{empirically computed} for practically relevant design universes. For example, one can $a)$ simulate welfare under expanded access $\alpha \rightarrow \alpha +\Delta$, allocating to more individuals, and $b)$ simulate the effects of enhanced interventions (\textit{e.g.} model the impact of a larger cash transfer under a specific choice of utility), or even $c)$ training under different distributions of feature sets. In Section~\ref{sec:poverty}, we compute welfare surfaces over data coverage, share of households eligible, and cash transfer size for poverty targeting in Ethiopia. In Section~\ref{sec: employment}, we do the same over institutional capacity and feature collection for the profiling of unemployed job seekers in Germany. As part of our work, we also provide open source software\footnote{Available at \url{https://github.com/unai-fa/relative-value-of-prediction}.} than enables practitioners to compute the design surface for common sets $\mathcal{U}$.

This stands in contrast to earlier theoretical work, which analyzed the relative value of policy levers under stylized assumptions such as Gaussian data and linear utilities 
\citep{pmlr-v235-perdomo24a, FISCHERABAIGAR2024101976}. Our approach lets planners ask a broader set of questions directly in their own problem settings, without first checking whether their data matches these prior stylized settings.

\subsection{Navigating the Meta-Design Universe}
Given a design space of feasible allocation problems $\mathcal{U}$, there are a number of questions that a social planner needs to answer. These include:
\begin{enumerate}[label={(Q\arabic*)}, left=1mm]
\item\label{q1} Given my current system, what kinds of local improvements are most effective?
\item\label{q2} Is my current configuration optimal for my budget?
\item\label{q3} How should I choose parameters optimally? 
\end{enumerate}
Given the framework we have set up so far, we now show how these questions can be answered by inspecting the welfare surface. As a means of presentation, we illustrate an overview of this methodology in the context of the theoretical model presented in \citet{FISCHERABAIGAR2024101976}, showing how it generalizes their analysis. They consider allocation systems that aim to identify individuals in need (see Example~\ref{worst-off-example}), and model outcomes $y$ and predictions $\hat{y}$ as following a bivariate normal distribution $(y,\hat{y}) \sim \mathcal{D}_\gamma$ with correlation $\gamma$. 
The design axes they analyze are the capacity of the allocation system $\alpha$ and stylized improvements in prediction by varying the squared-correlation (or $r^2$ of the predictions) $\gamma^2$, implying a design universe $\mathcal{U} = \{\mathcal{D}_\gamma, u, \alpha \from \gamma \in [0, 1], \ \alpha \in [0, 1] )$.  

More generally, when the design universe admits a suitable parametrization, we can index each design configuration $(\mathcal{D}, u, \alpha)$ by a vector $\theta \in \Theta=\Theta_1\times\cdots\times\Theta_d \subseteq \mathbb{R}^d_{\geq 0}$ where each coordinate $\theta_i\in\Theta_i\subseteq\mathbb{R}_{\geq0}$ corresponds to one design axis. In \cite{pmlr-v267-fischer-abaigar25a}, $\theta=(\alpha,\gamma^2) \in [0,1]^2$ and $V_*(\theta) := V_*(\alpha, \gamma^2) = V_*(\mathcal{D}_{\gamma}, u, \alpha)$ for a fixed utility $u(y, a) = a \cdot \1\{y \geq t_\beta\}$. Figure~\ref{fig:welfare-surface} shows the corresponding welfare surface. In Section~\ref{sec:casestudies}, we will instantiate our framework empirically, computing welfare surfaces from real data and exploring more realistic improvement axes.

\paragraph{Local Improvements \ref{q1}.} Given a parameterization $\theta$ of the design universe $\mathcal{U}$, the gradients of the welfare surface encode the magnitudes of incremental improvements along various axes, 
\begin{align*}
 \nabla V_*(\alpha, \gamma^2) = \left(   \frac{\partial}{\partial \alpha} V_*(\alpha, \gamma^2),  \frac{\partial}{\partial \gamma^2} V_*(\alpha, \gamma^2) \right)^\top   
\end{align*}
The two components capture the marginal value of expanding access and the marginal value of improving prediction, respectively.

As an analytical tool, the gradient of the welfare surface enables planners to directly visualize, and communicate to others, which design parameters are most important at a particular configuration. In particular, as seen in Figure~\ref{fig:welfare-surface-a}, the gradient of the welfare surface evaluated at a particular point is perpendicular to the contour line of that point. Vertical gradients, such as those on the bottom right of the plot indicate that a small improvement in prediction are relatively much more valuable than small improvements in access. Horizontal gradients, such as those on the left side, indicate the opposite: the marginal benefit of expanding access far outweighs that of expanding prediction. 

Our framing generalizes the insights from prior theoretical work. In particular, \cite{pmlr-v235-perdomo24a,FISCHERABAIGAR2024101976} formalize the relative value of prediction in allocation problems via a notion they term the \emph{prediction-access ratio} (PAR), defined as, 
\begin{align*}
    \mathsf{PAR}(\alpha, \gamma^2) = \frac{\partial}{\partial \alpha} V_*(\alpha, \gamma^2) \;/\; \frac{\partial}{\partial \gamma^2} V_*(\alpha, \gamma^2)
\end{align*}
This is exactly equal to the ratio of the entries in the gradient of the welfare surface for this specific choice of a design universe $\mathcal{U}$. Their main results consist of showing that this ratio is very large for small values of $\alpha$: an insight immediately visible from the near-flat gradients on the left of Figure~\ref{fig:welfare-surface-a}, which imply that the marginal welfare value of expanding capacity far exceeds that of improving the predictor.

Relative to prior work, our framework extends this analysis in two ways: it applies to any pair of design axes rather than just prediction versus capacity, and it can be computed empirically on data directly corresponding to a user's domain rather than requiring modeling assumptions. In particular, the ratio of any two partial derivatives of $V_*$ 
defines a \emph{marginal rate of substitution}, 
\begin{align*}
\mathsf{MRS}_{ij}(\theta) = \frac{\partial}{\partial \theta_i} V_*(\theta) \;/\; \frac{\partial}{\partial \theta_j} V_*(\theta)
\end{align*}
In microeconomic theory, the MRS between $i$ and $j$ denotes the rate at which a consumer can exchange good $i$ for good $j$ without changing welfare \citep{mas1995microeconomic}. In allocations, the MRS between two policy levers (of which the PAR is just a special case) specifies the rate at which a planner is willing to tradeoff two design parameters (e.g. prediction versus access) to maintain the same performance. As we will now show, the MRS gives practioners insight into whether their systems are cost-optimal.

\paragraph{Inspecting Optimality \ref{q2}.} When a planner has a fixed budget $B$ and cost function $c: \Theta \rightarrow \mathbb{R}_{\geq0}$, mapping design parameters to real numbers, the welfare surface enables planners to test whether their local configurations are suboptimal. In particular, consider the problem of checking whether a particular parameter vector $\theta \in \Theta=\Theta_1\times\cdots\times\Theta_d$ is optimal for a given budget constraint. That is, we want to test the following: $\theta^* \in \argmax_{\theta \in \Theta,\; c(\theta) \leq B}  V_*(\theta)$.

Classical nonlinear optimization theory states that any solution $\theta^*$ to this problem must satisfy the following first-order conditions, whenever the KKT
conditions apply:

\begin{fact}
\label{lem:kkt}
Suppose $\Theta=[0,1]^d$, and let $V_*:\Theta\to\mathbb{R}$ and 
$c:\Theta\to\mathbb{R}_{\geq0}$ be continuously differentiable. Suppose that $\theta^*$ solves
\begin{align*}
\max_{\theta\in\Theta,\; c(\theta)\leq B} V_*(\theta)
\end{align*}
and suppose that a standard constraint qualification holds at $\theta^*$, as is the case here when $c$ is linear in $\theta$. Then there exists a budget multiplier $\lambda\geq0$, with $\lambda(c(\theta^*)-B)=0$, such that, for each coordinate $i\in\{1,\dots,d\}$, 
$\frac{\partial V_*}{\partial \theta_i}(\theta^*) =
\lambda \frac{\partial c}{\partial \theta_i}(\theta^*)$ if
$\theta_i^*\in(0,1)$,
$\frac{\partial V_*}{\partial \theta_i}(\theta^*) \geq
\lambda \frac{\partial c}{\partial \theta_i}(\theta^*)$ if
$\theta_i^*=1$, and
$\frac{\partial V_*}{\partial \theta_i}(\theta^*) \leq
\lambda \frac{\partial c}{\partial \theta_i}(\theta^*)$ if
$\theta_i^*=0$.
\end{fact}
An immediate implication is that, at an optimal design, the marginal rate of substitution between any two design axes that are not at their bounds must equal the ratio of their marginal costs.
\begin{corollary}
Let $\theta^*\in\Theta$ satisfy the conditions of Fact~\ref{lem:kkt}. Then for any pair of coordinates $(i,j)$ with $0<\theta_i^*<1$ and $0<\theta_j^*<1$ for which the MRS and the cost ratio are well-defined, we have 
\begin{align*}
    \mathsf{MRS}_{ij}(\theta^*) = \frac{\partial}{\partial \theta_i} c(\theta^*) \;/\; \frac{\partial}{\partial \theta_j} c(\theta^*)
\end{align*}
\end{corollary}
Figure~\ref{fig:welfare-surface-b} provides a visual illustration of this condition. Intuitively, if the marginal rate of substitution between two design axes is not equal to the marginal ratio of costs, then the planner is better off by reducing their investment along one dimension and reallocating it to another. For instance, in Figure~\ref{fig:welfare-surface-b}, the configuration at the top of the feasible region (blue) is locally inefficient; the  planner is better off investing less in prediction and expanding access. Doing so would not increase their costs, but significantly improve the overall performance of their system.

We remark that using our software toolkit, planners can input their own cost functions and check whether these optimality conditions hold in the context of their own planning problems. We illustrate this methodology in the context of various case studies in the next section.

\paragraph{Choosing Optimal Design Parameters \ref{q3}.} A key 
advantage of the welfare surface framework is that it cleanly separates two concerns: understanding how welfare varies as a function of design choices, and understanding what those choices cost. The welfare surface $V_*(\theta)$ is a property of the allocation problem alone and can be computed by practitioners on a fixed dataset and over a wide range of design universes $\mathcal{U}$, independently of any cost considerations. Budget constraints enter only as a second step, as a feasible region $\mathcal{U}_B = \{\theta \in \mathcal{U} : c(\theta) \leq B\}$ overlaid on top of the surface (see Figure~\ref{fig:welfare-surface-c}). Finding the optimal design then reduces to identifying the point in $\mathcal{U}_B$ that lies on the highest iso-welfare contour, a problem addressed by our software and which is efficiently solvable due to the often low-dimensional structure of $\mathcal{U}$.

This separation is practically valuable. Beyond testing for local optimality as in Q2, a planner can compute the welfare surface once and then evaluate it under many different cost structures or budget levels. As the budget $B$ grows, the feasible region expands and the optimal configuration traces an \emph{expansion path} through $\mathcal{U}$ (see Figure 
\ref{fig:welfare-surface-d}). The shape of this path reveals how the 
optimal mix of investments shifts with scale --- for instance, whether data collection should precede capacity expansion at low budgets, or whether the two should be pursued in parallel. 

\section{Empirical Case Studies}
\label{sec:casestudies}

As demonstrated, the geometry of the welfare surface encodes answers to a range of design questions that prior analytical work has only addressed in stylized settings. We now turn to practice and compute the welfare surface on real-world data in two case studies, illustrating the kinds of conclusions a planner can draw and how those conclusions shift as the problem specification changes. 

\subsection{Poverty Targeting in Ethiopia}
\label{sec:poverty}
Direct cash transfer programs are globally ubiquitous, reaching 1.36 billion people across 962 programs in 2020 \citep{gentilini2022cash}. Many are targeted using proxy means tests (PMTs) \citep{barrientos2018social}, which use a simple predictive model to infer household poverty and prioritize the poorest of those for cash transfers \citep{grosh1995proxy, brown2018poor}. While PMT accuracy has been assessed in several contexts \citep{brown2018poor, schnitzer2024targeting, mcbride2018retooling}, the trade-offs between investing in prediction and allocation in this setting have not previously been analyzed. See Appendix~\ref{sec:app poverty} for additional background.

\paragraph{Allocation Problem.} A social planner wants to distribute cash transfers to poor households but cannot directly observe household poverty. Instead, short ``poverty scorecard'' surveys are collected in the field and a predictive model is trained to estimate household poverty from the survey responses. Households with the highest estimated poverty are then eligible for transfers up to a capacity threshold, $\pi(x) = \mathbf{1}\{f(x) \leq t_\alpha\}$. We simulate this approach using Ethiopia's 2015 Living Standards Measurement survey, covering 4,694 households, which contains both poverty scorecard questions (the features $x$) and measure of household consumption (serving as ground-truth poverty $y$).

\paragraph{Design Space.}  Because the scorecard surveys used to estimate household poverty in PMTs are expensive to conduct --- with a median per-survey cost of \$13 USD PPP reported in the literature \citep{aiken2023moving} --- such PMT registry databases often fail to cover all households in a country, or even all households in the poorest parts of a country. A recent review of PMT registries in low- and middle-income countries found a median registry coverage of 21\% of households (ranging from less than 1\% in Belize to over 99\% in Argentina) \citep{grosh2022revisiting}. Households without data cannot be individually targeted, effectively dramatically eroding accuracy of the targeting model at deployment time\footnote{To simulate households without scorecard data being excluded from individual targeting, in our setting we assign these households the population-mean, corresponding to the planner having no individualized information about them. Appendix~\ref{app: test-time} considers an alternative implementation in which households without data are treated as ineligible for the program and placed at the end of the allocation queue.}. The administrators of PMT-targeted cash transfer programs thus face a trade-off in where to invest resources: they can allocate resources to cash transfers --- serving a larger share of households or increasing the amount each recipient receives --- or they can expand data collection (that is, measure features for more people), ensuring that more households have data available for prediction. 

\begin{figure}[t]
\centering
\subfigure[Local Improvements \ref{q1} \label{fig:step-a}]{\includegraphics[width=0.32\textwidth]{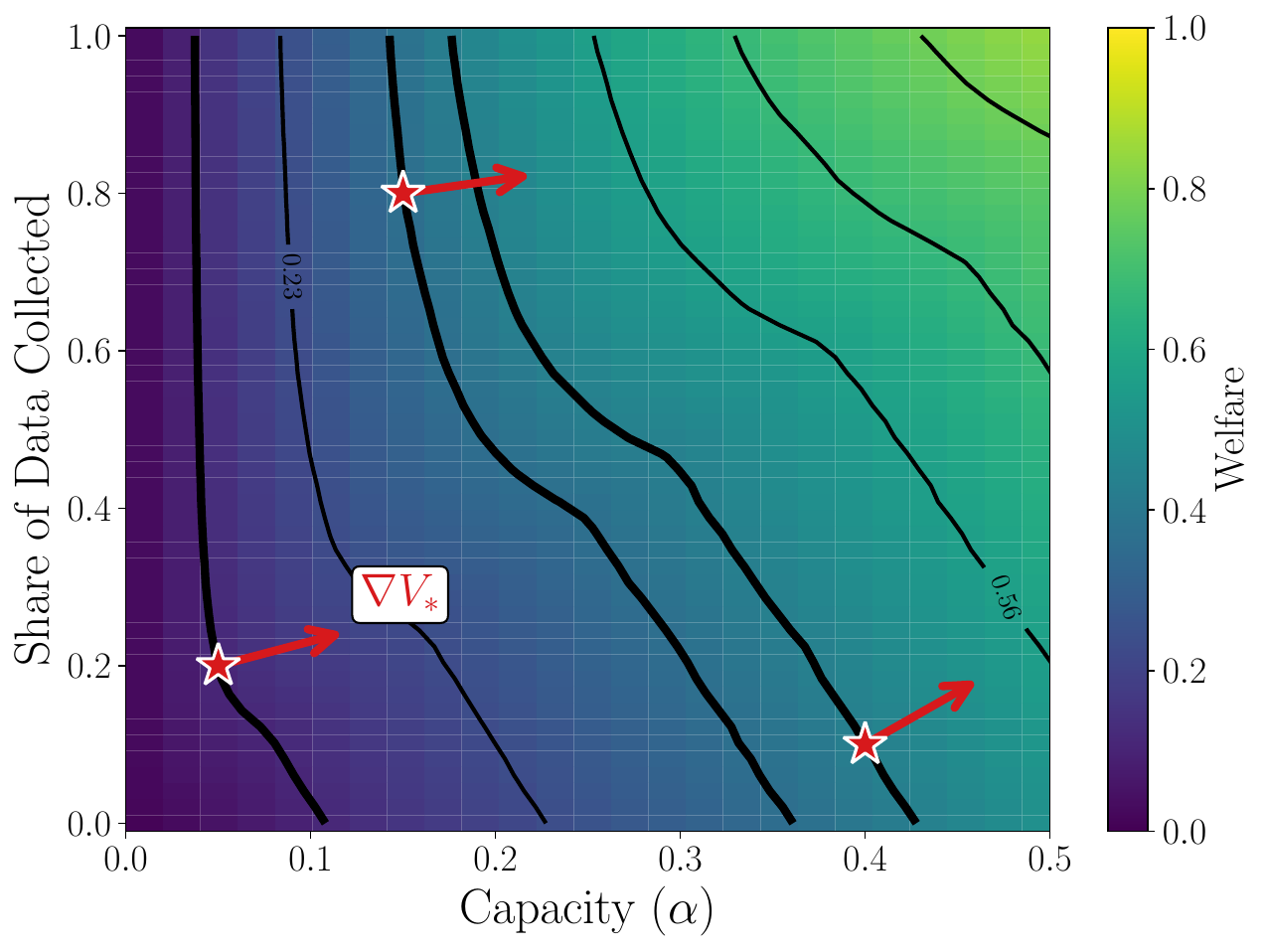}}
\subfigure[Local Optimality \ref{q2}\label{fig:step-b}]
{\includegraphics[width=0.32\textwidth]{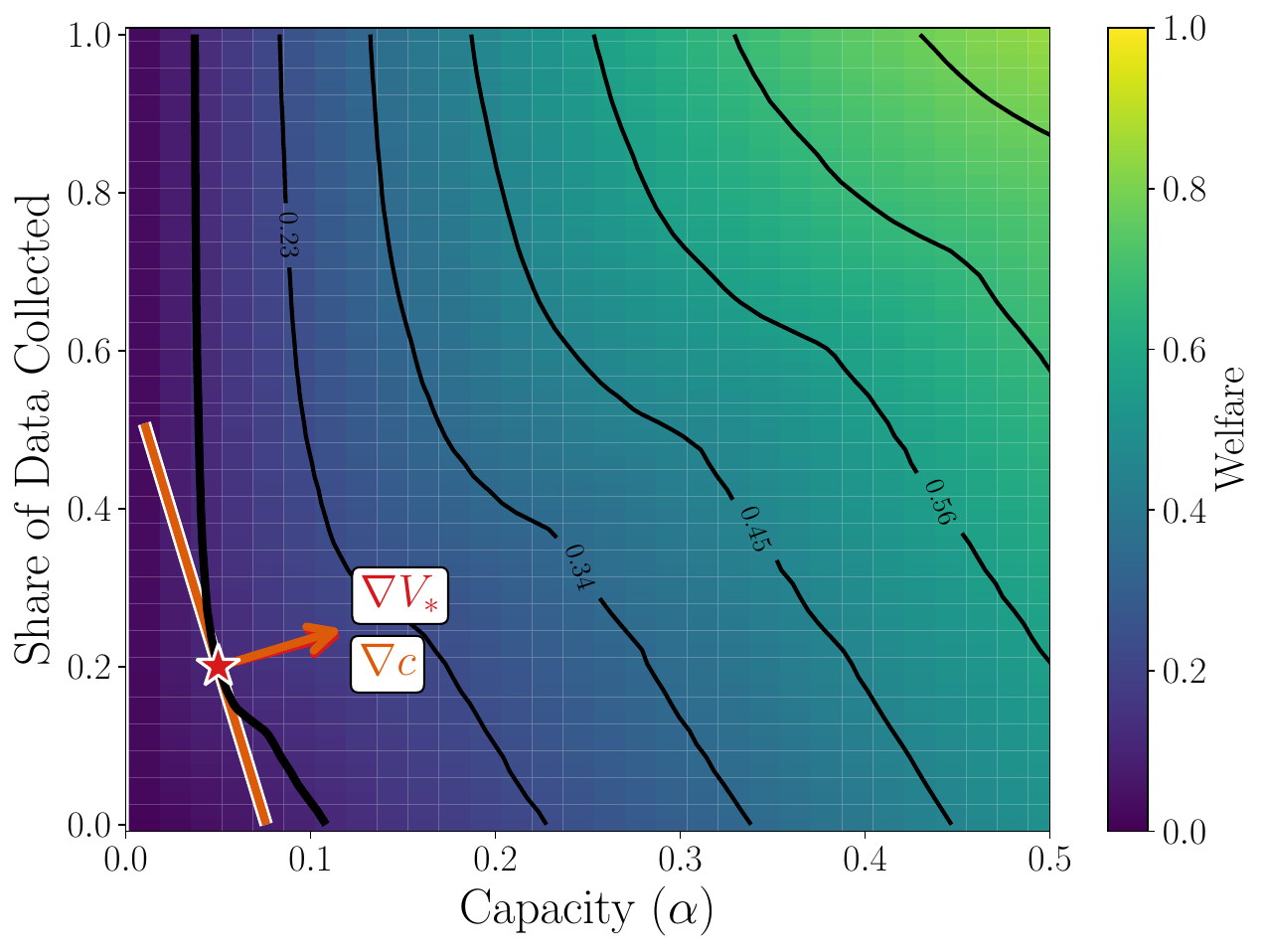}}
\subfigure[Expansion Path \ref{q3} \label{fig:step-c}]
{\includegraphics[width=0.32\textwidth]{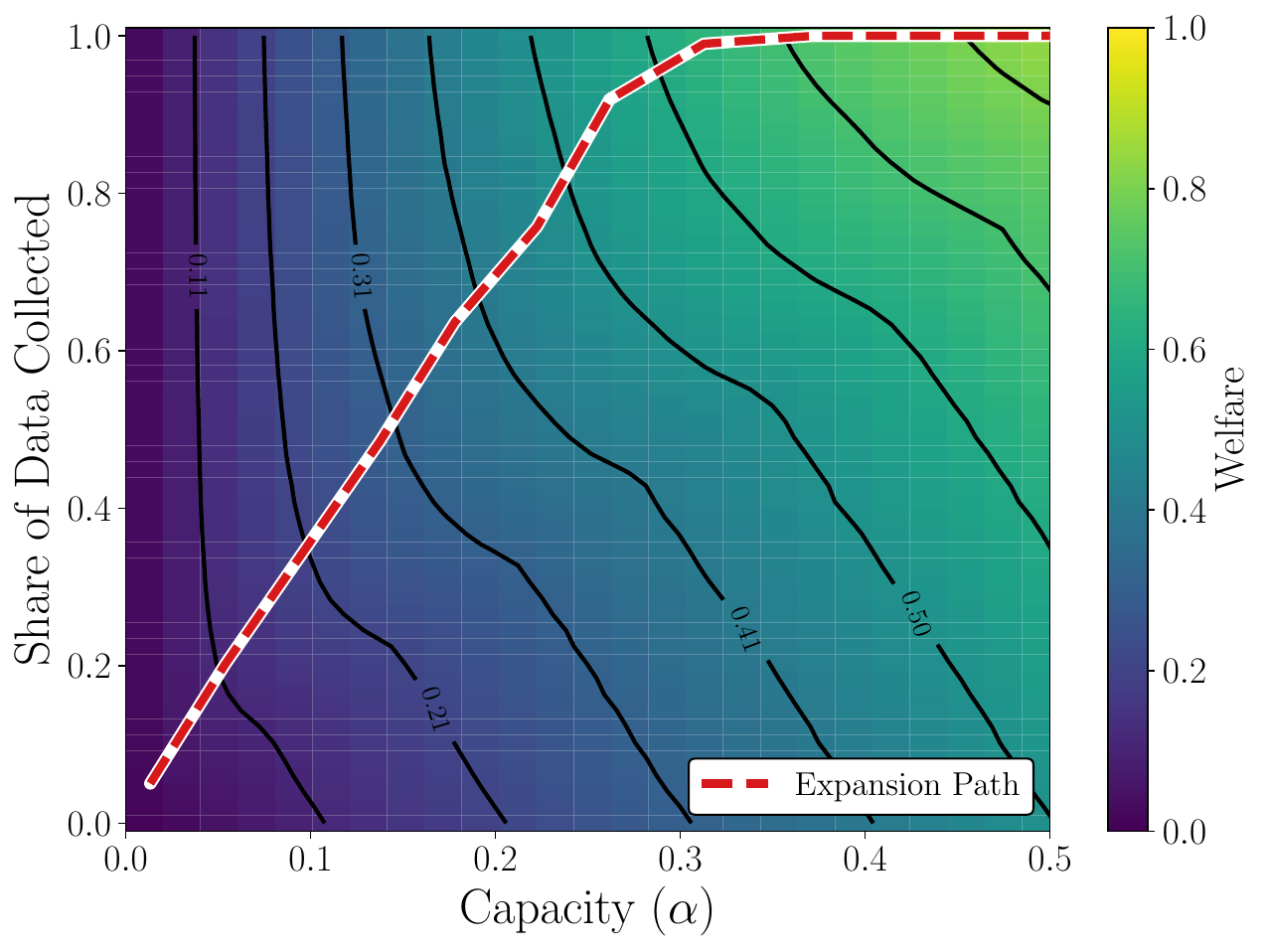}}
\caption{\textbf{Case Study 1.} Welfare surface for a hypothetical poverty targeting program in Ethiopia across the share of data collected and the capacity of the program. Welfare measured as the number of correctly identified individuals below the poverty line (33\%). Averaged over $100$ cross validation splits. Data collection refers to scorecard surveys at a median cost of \$13 USD PPP per household; capacity investment expands the share of households receiving transfers of \$100 USD PPP.}
\label{fig:step}
\end{figure}

\subsubsection{Case Study 1: Identifying Poor Households}
A standard objective is to maximize the number of poor households reached by the program,
\begin{align*}
  u(y, a) = a \cdot \mathbf{1}\{y \leq \bar{y}\}  
\end{align*}
where $\bar{y}$ denotes the poverty line and $a \in \{0,1\}$ denotes the targeting decision. Using the international poverty line of \$3 per day, roughly 33\% of households fell below this threshold in 2015 \citep{worldbank2026pip}. For now, we assume a benefit size of \$100 USD PPP (purchasing power parity), corresponding to an 8\% increase in yearly consumption for the average household in 2015.

\textbf{\ref{q1}} Inspecting the local rate of improvement for this design universe (Figure~\ref{fig:step-a}), we find that a small improvement in capacity tends to far outweigh the impact of a small increase in the share of data collected. For a planner deciding where to invest at the margin, welfare is far more responsive to capacity expansion than to data collection in this setting. \textbf{\ref{q2}}  To decide whether a program configuration is optimal, we also need to account for costs. Given the survey cost of \$13 PPP per household and a transfer size of \$100 PPP \citep{aiken2023moving}, capacity expansion is roughly eight times more expensive per household than data collection. The natural question is whether a typical anti-poverty program --- calibrated here at 20\% data coverage \citep{grosh2022revisiting} and 5\% capacity\footnote{Ethiopia's overall social protection coverage around 2015 was 21\% \citep{worldbank_aspire}, but the coverage of one cash transfer program is likely to be substantially below this overall coverage level.} --- strikes the right balance between the two. We find that it does, meaning the higher marginal welfare of capacity is exactly offset by its higher cost (Figure~\ref{fig:step-b}), meaning that the planner is spending the budget cost-efficiently. \textbf{\ref{q3}} We then optimize over the design space with budgets up to \$2.7 billion PPP per year\footnote{This maximum budget corresponds to 100\% capacity for a $\$100$ per year cash transfer program covering Ethiopia's 27.3 million households.}, tracing the expansion path that characterizes the optimal budget allocation. We find that data collection and capacity expansion should go hand in hand as the budget grows.

We find that while welfare is far more responsive to capacity than to data collection, the cost structure implies that both investments matter when building a program. Importantly, a planner operating at small capacity should not focus solely on collecting data to optimally allocate a small transfer budget, but should also invest in expanding the program itself. As we show in the next case study, however, these conclusions can change under alternative specifications of the utility function.

\begin{figure}[t]
\centering
\subfigure[Local Improvements \ref{q1} \label{fig:crra-a}]{\includegraphics[width=0.32\textwidth]{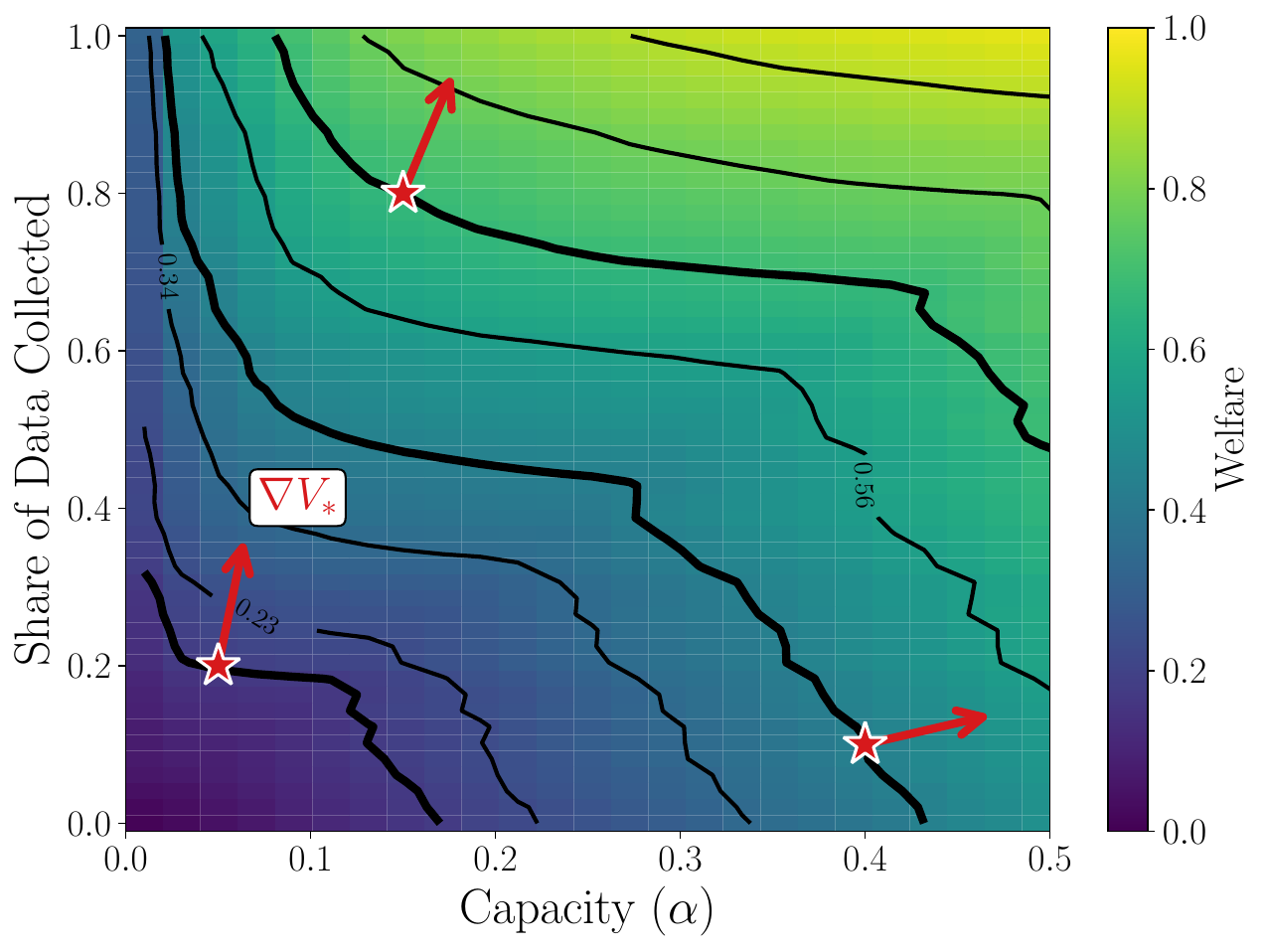}}
\subfigure[Local Optimality \ref{q2}\label{fig:crra-b}]
{\includegraphics[width=0.32\textwidth]{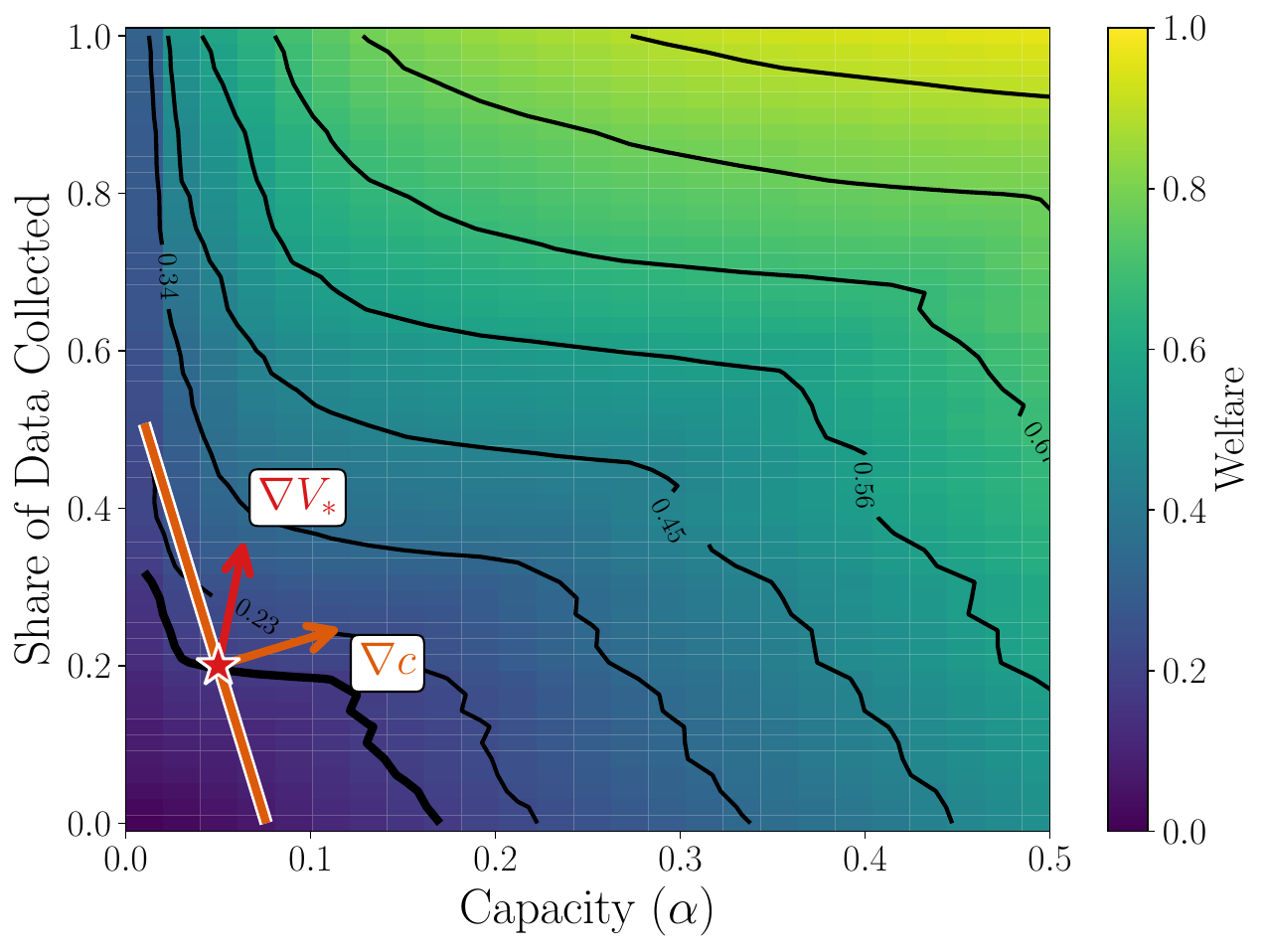}}
\subfigure[Expansion Path \ref{q3} \label{fig:crra-c}]
{\includegraphics[width=0.32\textwidth]{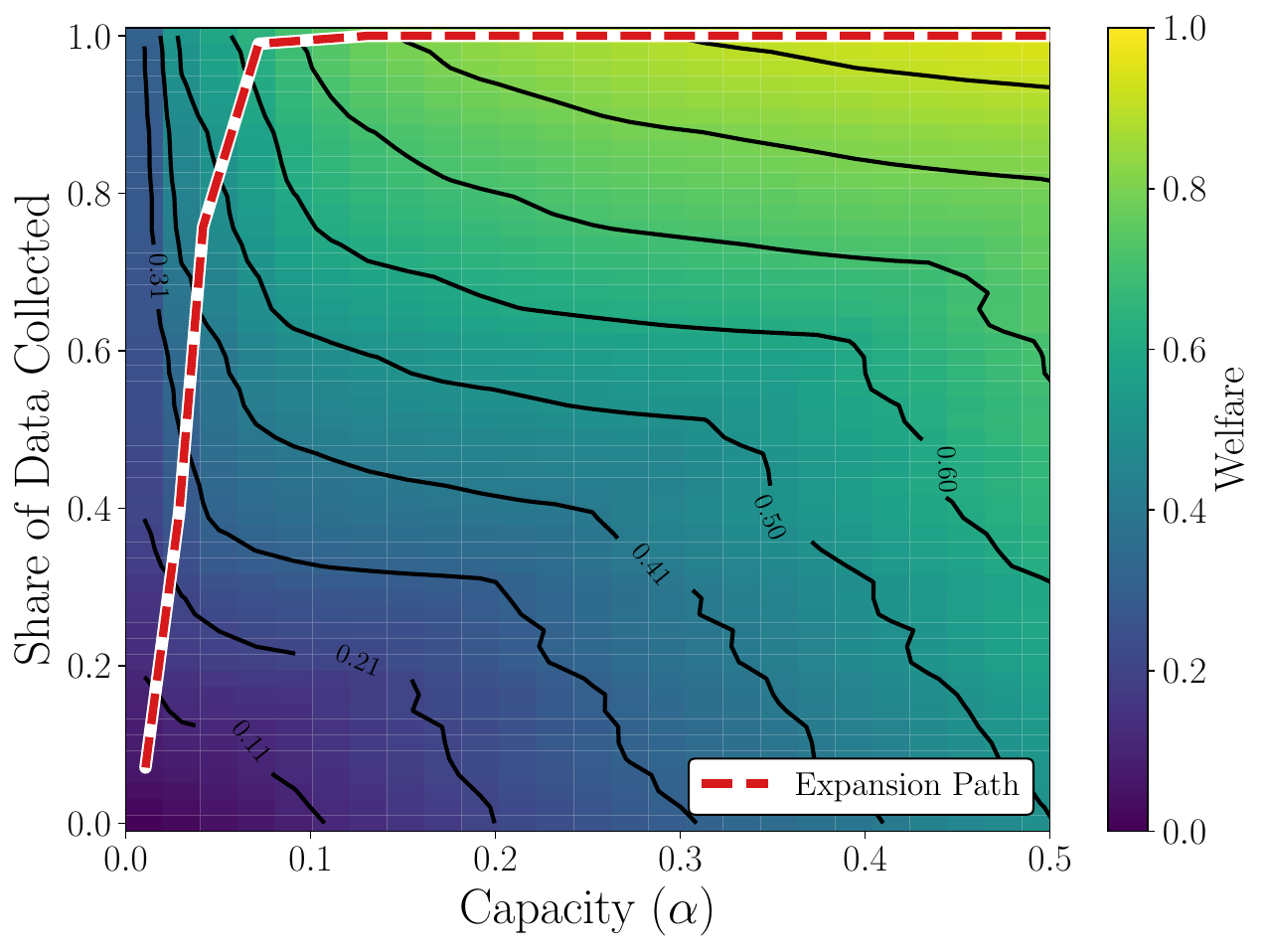}}
\caption{\textbf{Case Study 2.} Welfare surface for a hypothetical poverty targeting program in Ethiopia across the share of data collected and the capacity of the program. Welfare measured as gains under CRRA utility. Data collection refers to scorecard surveys at a cost of \$13 USD PPP per household; capacity investment expands the share of households receiving transfers of \$100 USD PPP.}
\label{fig:crra}
\end{figure}

\subsubsection{Case Study 2: CRRA Utility}
While evaluating the number of correctly identified poor households is common practice in poverty targeting, it assumes that prioritizing poorer households over less poor ones has no additional benefit, and that providing transfers to households above the poverty line has no benefit at all. The CRRA utility function \citep{hanna2018universal} is standard in development economics, and challenges this assumption by modeling diminishing marginal impacts of transfers as a function of household welfare, 
\begin{align*}
    u_b(w,a)
 = \frac{(w+ba)^{1-\rho}-w^{1-\rho}}{1-\rho}, \quad\rho>0, \; \rho\neq1
\end{align*}
where $b$ is the size of the cash transfer and higher $\rho$ implies larger relative gains from targeting poorer households (see Figure~\ref{fig:crra-utility-app} in Appendix). Following \citet{hanna2018universal},  we use a value of $\rho=3$ in our empirical simulations, within the standard range 2-4 \citep{gandelman2015risk}. 

\textbf{\ref{q1}} Switching to a CRRA utility objective flips the conclusion from Case Study 1. Across larger regions of the design space, a small improvement in prediction now has a far larger impact on welfare than a small improvement in capacity (Figure~\ref{fig:crra-a}). Relative to the earlier step function, CRRA utility rewards finer distinctions among households both below and above the poverty line, making accurate targeting more valuable. \textbf{\ref{q2}} The same configuration as in Case Study 1 is no longer locally cost-efficient (Figure~\ref{fig:crra-b}). The planner could now improve welfare at the same total cost by rebalancing toward additional data collection and spending less on program capacity. \textbf{\ref{q3}} Prediction and capacity expansion no longer go hand in hand. The expansion path under the same cost structure (Figure~\ref{fig:crra-c}) shows that much of the initial budget should be spent directly on data collection.

Even though both objectives prioritize transfers to poor households, they have drastically different welfare surfaces. They flip the planner's optimal strategy from expanding capacity and data collection together to prioritizing data collection first. This points to a perhaps underappreciated lesson: choices about how to evaluate program outcomes, often made somewhat arbitrarily in practice, can imply fundamentally different design decisions and different conclusions about the value of prediction.

\begin{figure}[t]
\centering
\subfigure[Local Improvements \ref{q1} \label{fig:benefit-a}]{\includegraphics[width=0.32\textwidth]{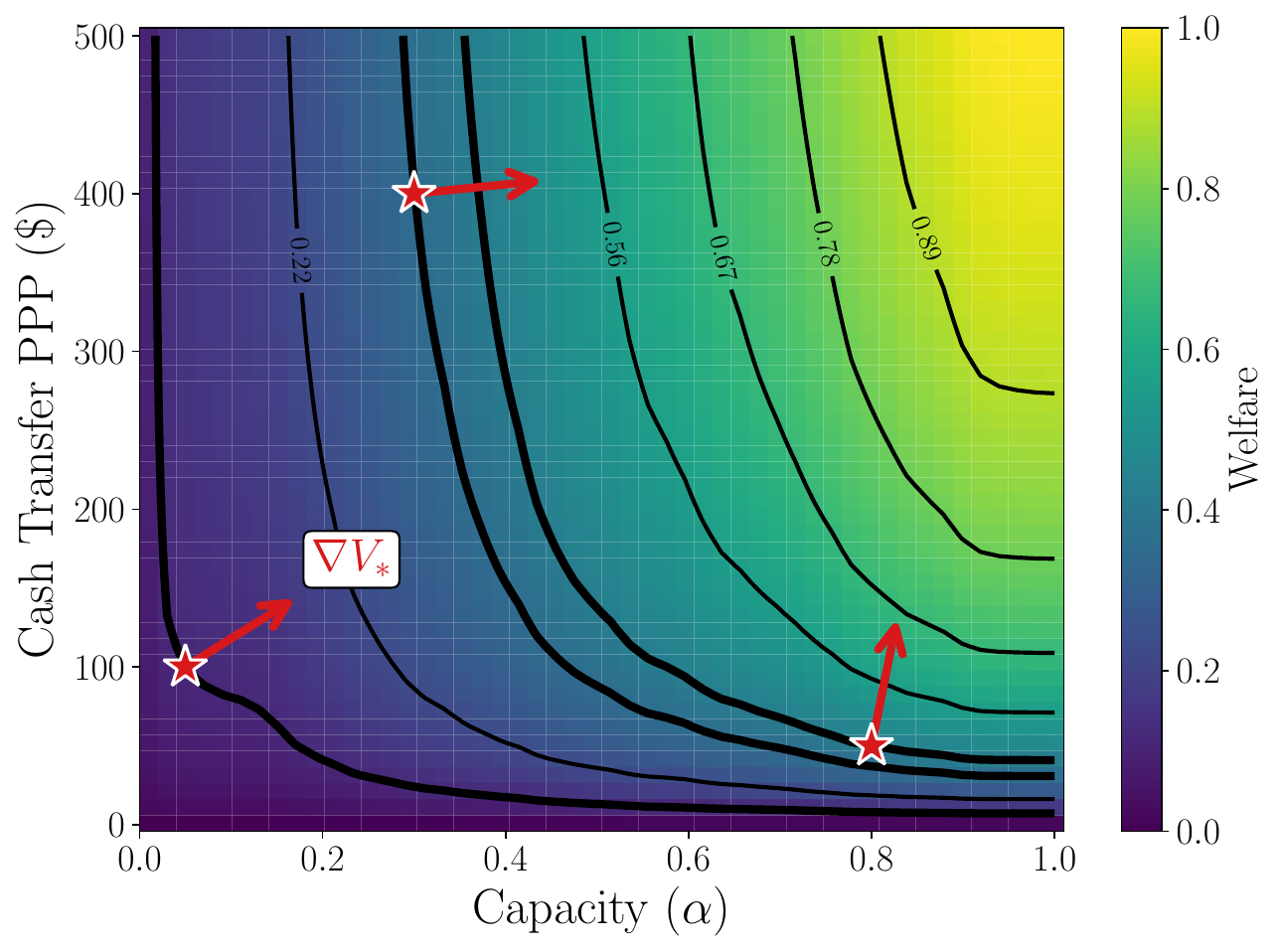}}
\subfigure[Local Optimality \ref{q2}\label{fig:benefit-b}]
{\includegraphics[width=0.32\textwidth]{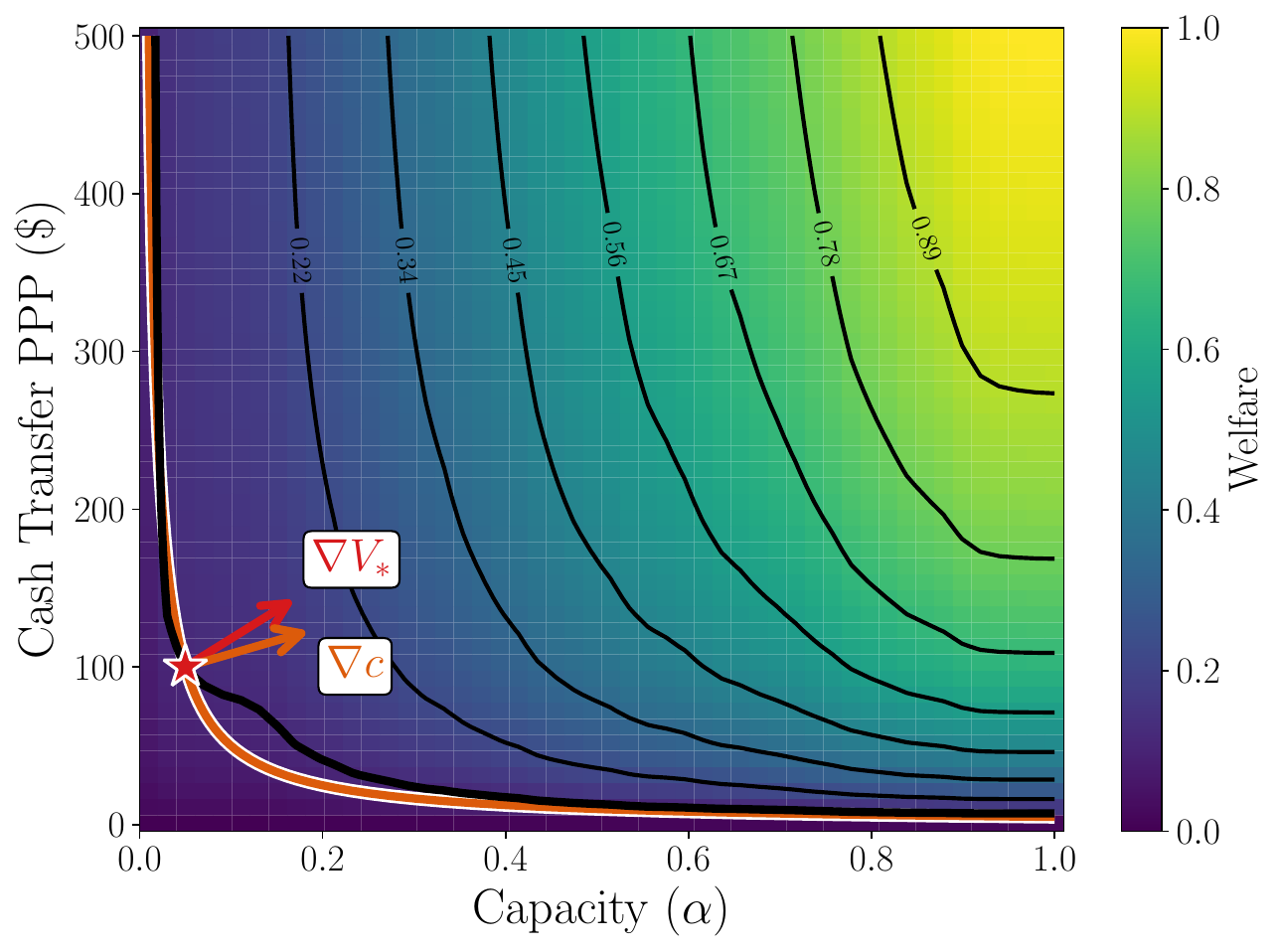}}
\subfigure[Expansion Path \ref{q3} \label{fig:benefit-c}]
{\includegraphics[width=0.32\textwidth]{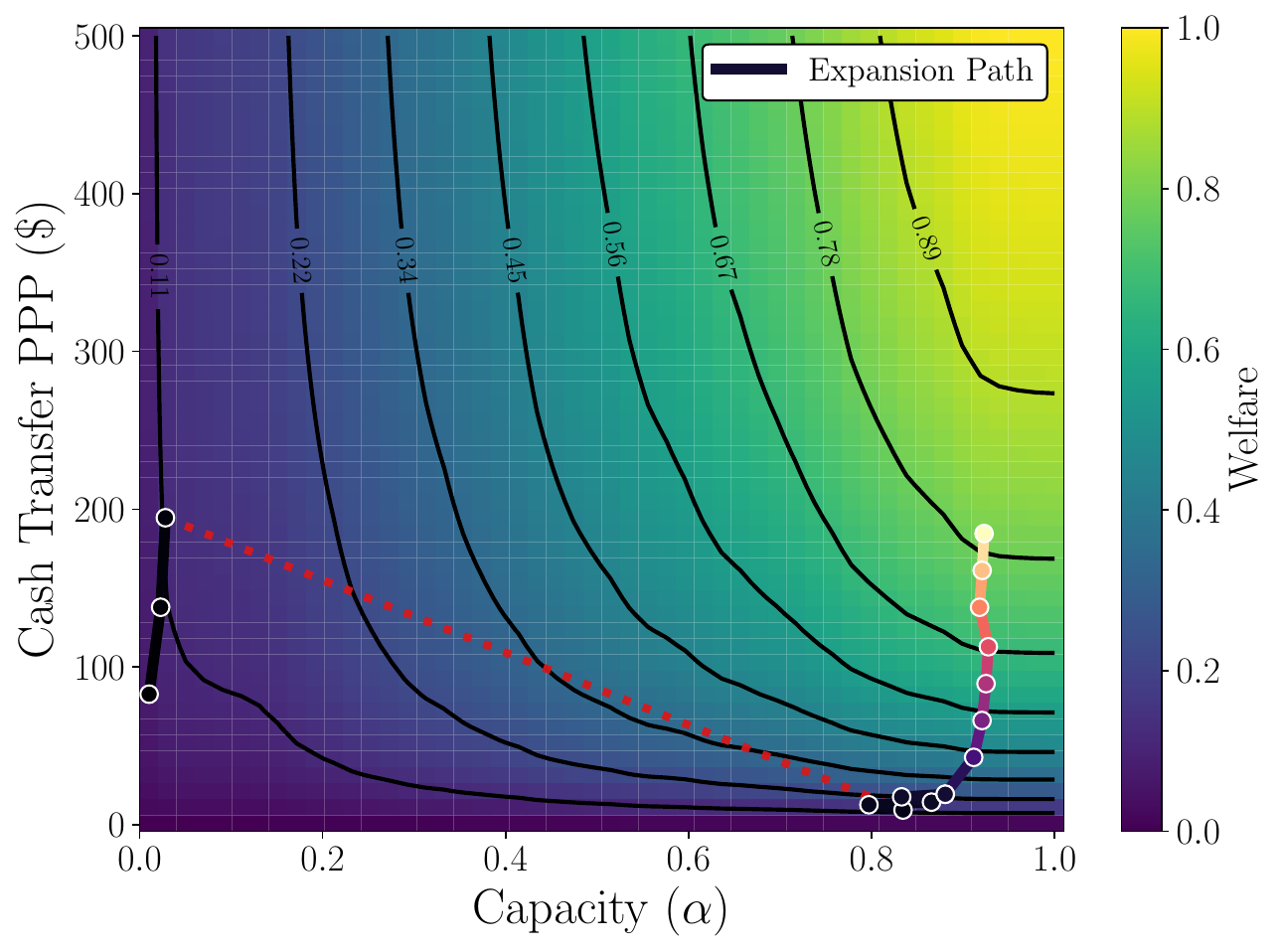}}
\caption{\textbf{Case Study 3.} Welfare surface for a hypothetical poverty targeting program in Ethiopia across transfer size and the capacity of the program. Welfare measured as gains under CRRA utility. Data collection is held fixed  at $20\%$. The dotted red line in (c) indicates a discontinuous jump in the expansion path.}
\label{fig:benefit}
\end{figure}

\subsubsection{Case Study 3: Varying Transfer Sizes}

So far we fixed the transfer size at \$100 PPP per beneficiary. But a planner may also face a choice between writing a small number of large checks or a large number of small checks. We consider trade-offs between the intensive margin (transfer amount $b$) and the extensive margin (capacity $\alpha$).

\textbf{\ref{q1}} Both capacity and transfer size can have large marginal impacts on welfare, and they are strongly complementary (Figure~\ref{fig:benefit-a} at $20\%$ data coverage). The marginal gain from increasing transfer size grows with program capacity, and vice versa. \textbf{\ref{q2}} Costs are more complex in this design space, since serving more households is more expensive when transfers are larger. This creates an interesting dynamic: in regimes where one margin has a much larger welfare impact, it also tends to be much more expensive, making the tradeoff far less obvious than it might first appear. At the typical program configuration, for instance, the planner could improve welfare by targeting a smaller population with larger transfers (Figure~\ref{fig:benefit-b}). \textbf{\ref{q3}} Overall, the expansion path has a discontinuity; the planner should switch between two qualitatively different strategies. At lower budgets, they should favor larger transfers to a smaller group of recipients, while at higher budgets, it becomes more valuable to provide smaller transfers to a broader share of the population. 

As we saw, even a seemingly simple addition to the design space can lead to a non-trivial finding; the planner should qualitatively switch how transfers are distributed as the budget grows. But expanding the design space further to include all three levers changes the answer again, with the discontinuity disappearing and the expansion path more smoothly balancing capacity and transfer size (Figure~\ref{fig:benefit-capacity-label} in the Appendix). Our tools enable planners to reliably make these decisions in their own domains.

\begin{figure}[t]
\centering
\subfigure[Local Improvements \ref{q1} \label{fig:emp-a}]{\includegraphics[width=0.32\textwidth]{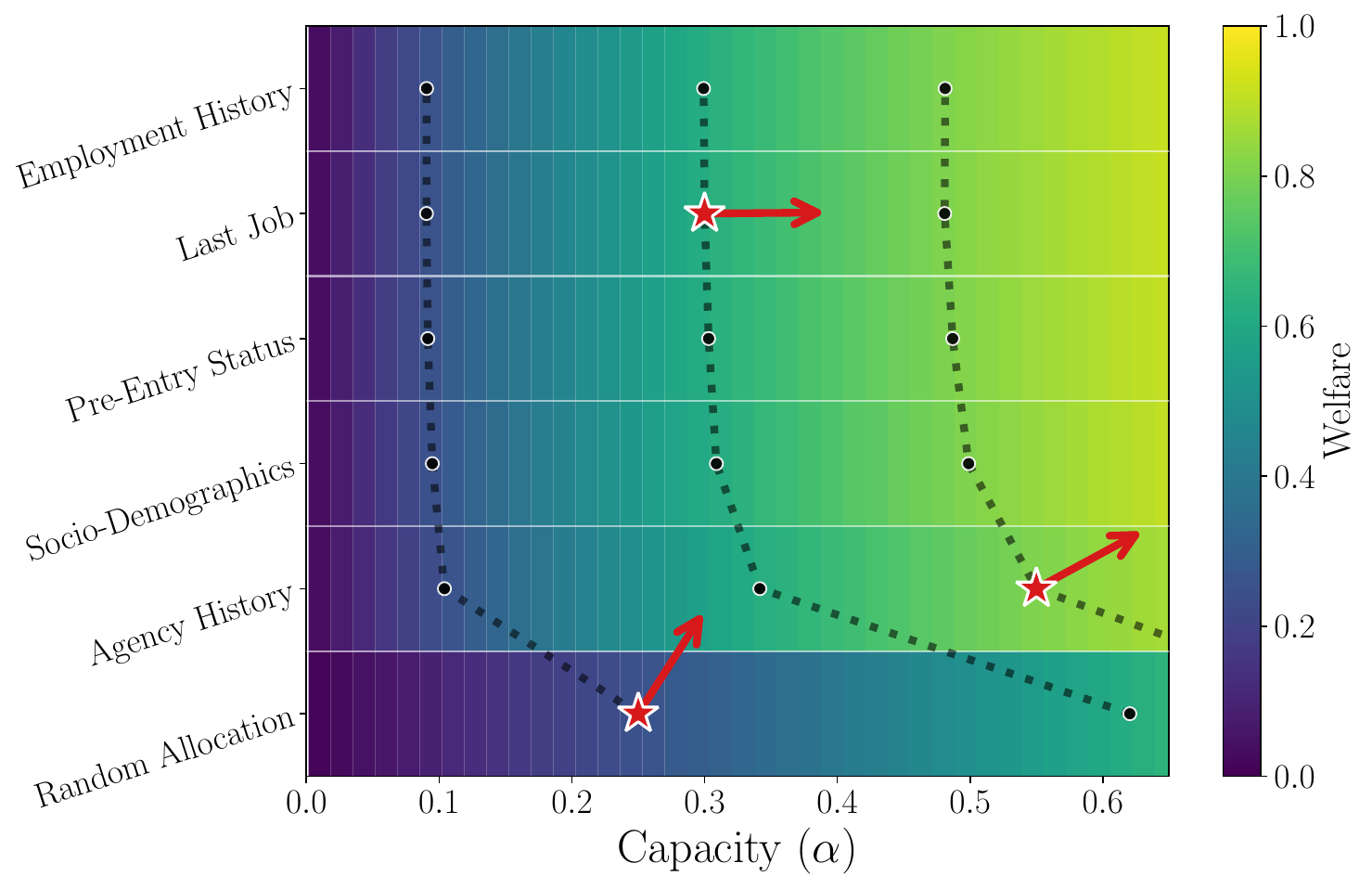}}
\subfigure[Local Optimality \ref{q2}\label{fig:emp-b}]
{\includegraphics[width=0.32\textwidth]{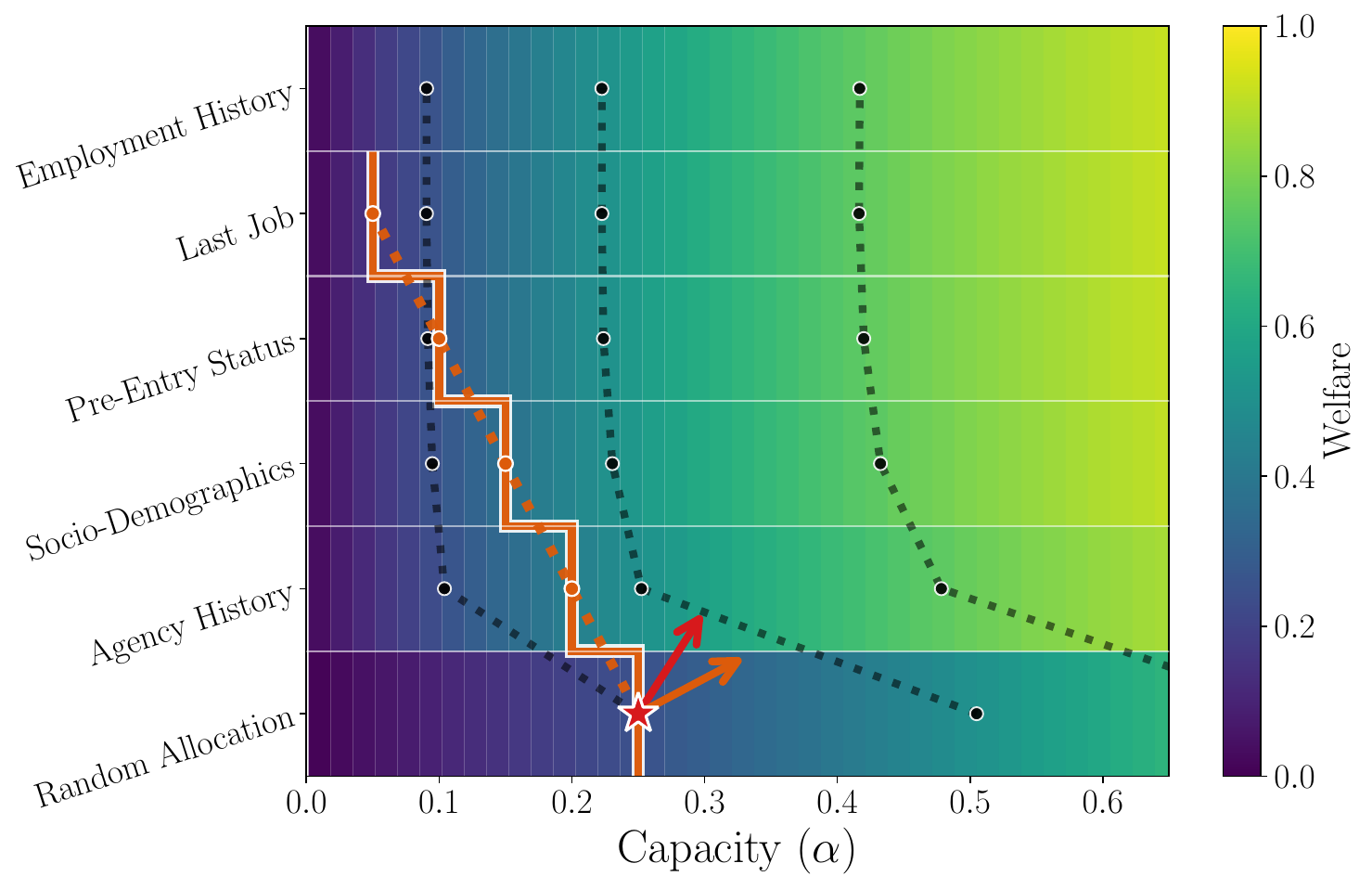}}
\subfigure[Expansion Path \ref{q3} \label{fig:emp-c}]
{\includegraphics[width=0.32\textwidth]{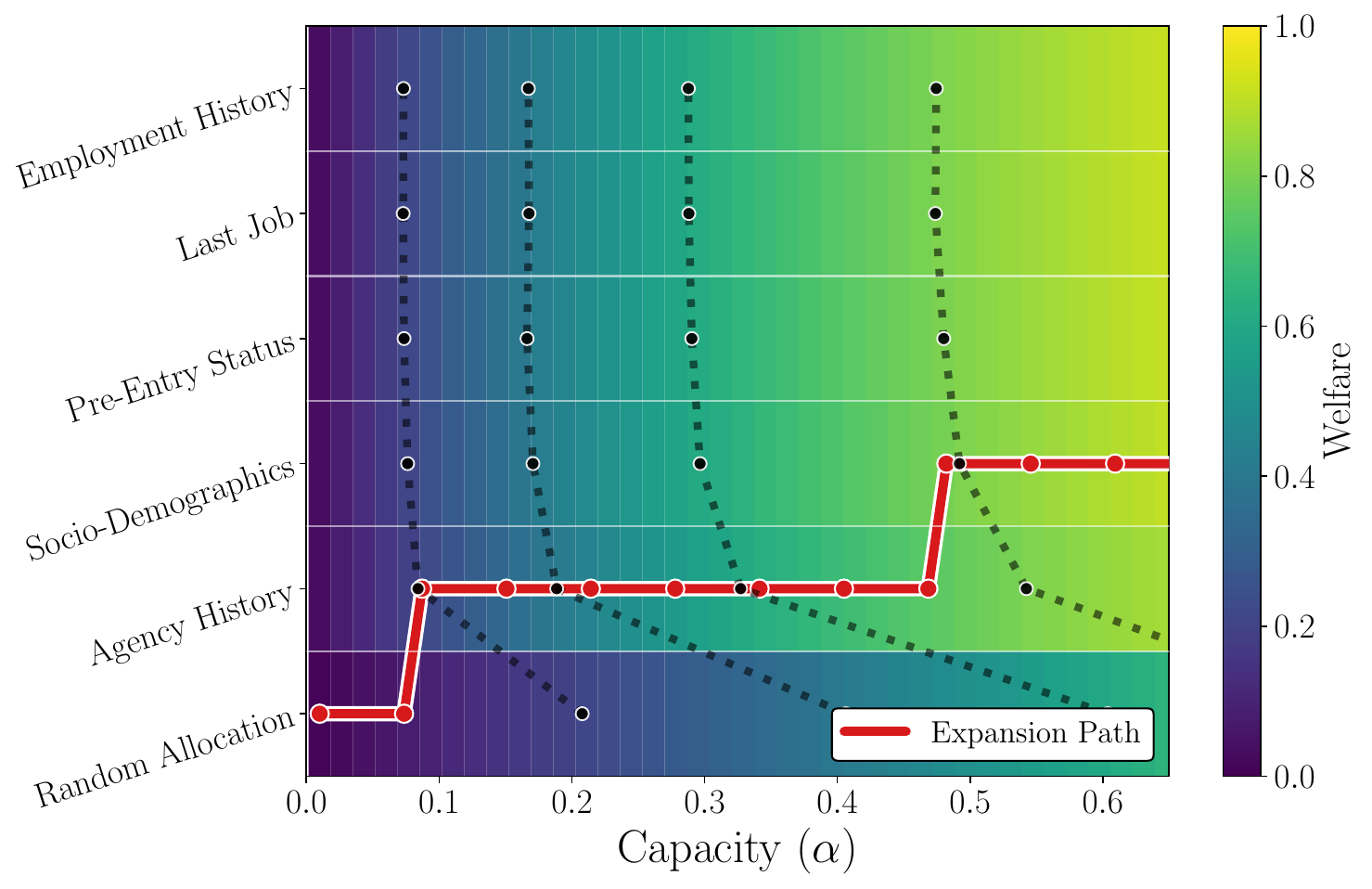}}
\caption{\textbf{Case Study 4.} Welfare surface for a hypothetical job seeker profiling system. The ordering along the feature axis is determined by greedy forward selection, cumulatively adding the group with the largest marginal welfare gain at each step. At each configuration, arrows encode the capacity increase that would be required to achieve the same welfare gain as adding the next feature group.}
\label{fig:employ}
\end{figure}

\subsection{Profiling of Job seekers in Germany}
\label{sec: employment}

\begin{figure}[t]
\centering
\includegraphics[width=0.5\textwidth]{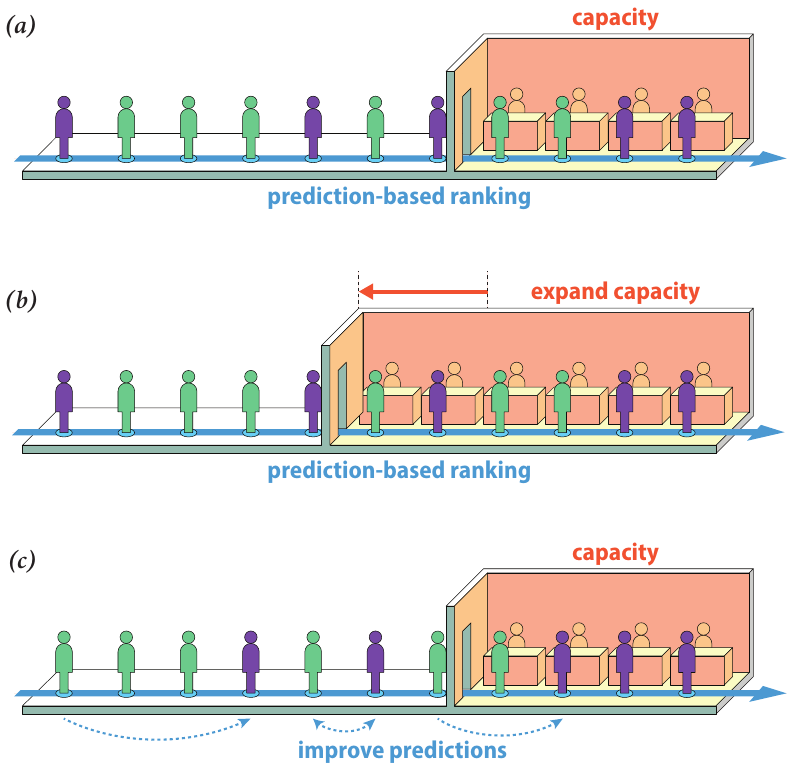}
\caption{\textbf{(a)} Illustrates the task of identifying jobseekers at risk of long-term unemployment (purple) under a fixed capacity constraint. The employment office observes only imperfect predictions of risk, so the ranking of individuals is uncertain. \textbf{(b)} Increasing capacity expands the set of individuals who can be served. \textbf{(c)} Improving prediction sharpens the ranking, allowing limited resources to be targeted more effectively toward those truly at risk.}
\label{fig:motivation}
\end{figure}

Public Employment Services (PES) across countries deploy predictive models to support one of their most critical tasks: the allocation of support programs to job seekers at risk of long-term unemployment (LTU) \citep{kortner2023predictive, desiere2019statistical}. LTU incurs considerable costs for public welfare, and the support programs allocated to at-risk job seekers account for large shares of PES spending. While prediction-based profiling of job seekers is a topic of substantial policy relevance, to our knowledge the meta-design of these systems has not been systematically analyzed before. See Appendix~\ref{app: emply data} for additional information.

\paragraph{Allocation Problem.} Profiling tools are commonly used to predict which job seekers are at risk of long-term unemployment (12+ months in Germany, \citep{bach2023impact}) based on pre-unemployment covariates $\mathcal{X}$ such as employment and benefits history, corresponding to the utility $u(y, a) = a \cdot \mathbf{1}\{y \geq t_\beta\}$ from Example~\ref{worst-off-example}. Those with the highest predicted risk are prioritized for support, i.e., $\pi(x) = \mathbf{1}\{f(x) \geq t_\alpha\}$. To evaluate a hypothetical profiling system, we secured access to a 2\% random sample of German administrative labor market records, comprising over 60 million entries of employment spells from all residents of Germany from 1970 to 2021. The Sample of Integrated Labour Market Biographies (SIAB) is provided via a Scientific Use File from the Research Data Centre (FDZ) of the German Federal Employment Agency (BA) at the Institute for Employment Research (IAB) \citep{SchmuckerVomBerge2023SIABR7521v1}.

\paragraph{Design Space.} From the perspective of a public employment office, we consider two design axes. The first is institutional capacity: how many job seekers can be prioritized for additional caseworker attention. The second is the quality of predictions used for targeting. In employment services, where agencies already operate on large administrative datasets, improving prediction is typically not a matter of sample size but of collecting richer information on job seekers \citep{desiere2019statistical}. The dataset used in our experiments was itself constructed by linking administrative sources from the German employment system, but is not yet operationalized in agency workflows. Whether making such data available for profiling would justify the cost is a natural question; we analyze the welfare side of this tradeoff.

\subsubsection{Case Study 4: Identifying Long-Term Unemployment}
To operationalize the feature axis, we group available features into five categories (see Table~\ref{tab:feature-groups}). Each group contains features that broadly share a common data source and collection effort, ranging from basic socio-demographic information to longitudinal employment and agency records that require linking across administrative registries. 

\textbf{\ref{q1}} We find that most of the welfare gain from prediction comes from information on past interactions with the employment agency (Figure~\ref{fig:emp-a}). Additional feature groups contribute comparatively little, suggesting that a decision rule based on agency history alone may suffice. \textbf{\ref{q2}} In contrast to poverty targeting, exact costs are harder to pin down, as the cost of collecting or linking features depends on institutional specifics. For illustration, we assume each feature group costs the same as expanding capacity by five percentage points (Figure~\ref{fig:emp-b}). Under these costs, investing in the first feature group is clearly worthwhile at low capacity, but once the most informative groups are in place, cost-adjusted returns to further data collection are small. The welfare surface provides the basis for such comparisons once institutional costs are known. \textbf{\ref{q3}} Tracing an expansion path under this cost structure (Figure~\ref{fig:emp-c}), the optimal strategy invests early in features to avoid random allocation, then shifts toward expanding capacity.

Taken together, these results suggest that across the administrative data sources available to the agency, a (simple) decision rule based on past unemployment history captures nearly all the welfare-relevant signal. Whether qualitatively different data sources, such as structured interviews or questionnaires with job seekers \citep{desiere2019statistical}, could change this picture is something the same analysis could directly evaluate once such data becomes available.

\section{Conclusion}

In this paper, we develop a framework for navigating the meta-design of resource allocation problems. We define the questions a social planner may ask when designing a social program and develop empirical tools to answer them. We see this as a first step in a budding area of work on the value of prediction, with many exciting directions for future research.

While we apply our framework to two real-world settings, it could be used to study a much larger design space. For example, one could consider allocating multiple resources rather than making binary allocation decisions, or explore other utility functions and constraints, including fairness constraints. Future work could also extend the framework to dynamic allocation settings, revisiting the welfare surface at different points in time as resources are continuously allocated, and exploring connections to online learning. It would also be interesting to consider other cost structures, such as fixed costs and costs that amortize over time. More broadly, we see value in developing and testing this work in closer collaboration with social planners, guided by the questions they face and the information they need.

\paragraph{Limitations.} To apply our analysis, a social planner needs to have well-specified allocation goals, information on the costs of prediction, and an understanding of the costs and impacts of different policy levers. In many settings planners already have access to some or all of this information. Cost data, however, are not always collected for prediction: our work points to the importance of collecting high quality data on costs of prediction, similar to the cost data necessary to conduct cost-efficiency analysis in other domains \citep{petrou2011economic, mcewan2012cost}.  We also see an exploratory role for the framework, allowing planners to test conclusions across plausible ranges of parameters. This can help planners identify which design choices and assumptions are particularly important and require careful consideration in their own problem domain.


\section*{Acknowledgements}
We are thankful to Chris Hays and Sendhil Mullainathan for insightful comments and discussions. We gratefully acknowledge the Social Sciences Computing Facility (SSCF) at UC San Diego for providing computational resources. This research utilized the Social Sciences Research and Development Environment (SSRDE) cluster. JCP work is supported by an AI2050 Early Career Fellowship from Schmidt Sciences. UFA acknowledges the support by the DAAD programme Konrad Zuse Schools of Excellence in Artificial Intelligence, sponsored by the German Federal Ministry of Education. This work was partially conducted while UFA was a PhD student at the Department of Statistics at LMU Munich, and  a visiting student at MIT LIDS. We thank Nanina Föhr for helpful feedback and support with the design of the paper’s illustrations.

\bibliography{bib}

\appendix
\clearpage

\section{Proof of Fact~\ref{lem:kkt}}

\begin{proof}
Write the constraints in the form
\begin{align*}
c(\theta)-B\leq 0,\qquad \theta_i-1\leq 0,\qquad -\theta_i\leq 0.
\end{align*}
The full Lagrangian is
\begin{align*}
\widetilde{\mathcal{L}}(\theta,\lambda,\mu^+,\mu^-) = V_*(\theta)-\lambda\bigl(c(\theta)-B\bigr)-\sum_{i=1}^d \mu_i^+(\theta_i-1)-\sum_{i=1}^d \mu_i^-(-\theta_i),
\end{align*}
where $\lambda\geq0$ is the multiplier for the budget constraint, $\mu_i^+\geq0$ is the multiplier for the upper bound $\theta_i\leq1$, and $\mu_i^-\geq0$ is the multiplier for the lower bound $\theta_i\geq0$. By the KKT conditions, the multipliers satisfy complementary slackness:
\begin{align*}
\lambda\bigl(c(\theta^*)-B\bigr)&=0,\\
\mu_i^+(\theta_i^*-1)&=0 &&\text{for each } i,\\
\mu_i^- \theta_i^*&=0 &&\text{for each } i.
\end{align*}
They also satisfy stationarity:
\begin{align*}
\frac{\partial V_*}{\partial \theta_i}(\theta^*)-\lambda\frac{\partial c}{\partial \theta_i}(\theta^*)-\mu_i^+ +\mu_i^-= 0
\qquad\text{for each } i.
\end{align*}
Now fix a coordinate $i$. If $0<\theta_i^*<1$, then complementary slackness implies $\mu_i^+=\mu_i^-=0$, and stationarity gives
\begin{align*}
\frac{\partial V_*}{\partial \theta_i}(\theta^*)= \lambda\frac{\partial c}{\partial \theta_i}(\theta^*).
\end{align*}
If $\theta_i^*=1$, then complementary slackness implies $\mu_i^-=0$, and stationarity gives
\begin{align*}
\frac{\partial V_*}{\partial \theta_i}(\theta^*)-\lambda\frac{\partial c}{\partial \theta_i}(\theta^*)=\mu_i^+\geq0.
\end{align*}
Hence
\begin{align*}
\frac{\partial V_*}{\partial \theta_i}(\theta^*)\geq\lambda\frac{\partial c}{\partial \theta_i}(\theta^*).
\end{align*}
If $\theta_i^*=0$, then complementary slackness implies $\mu_i^+=0$, and stationarity gives
\begin{align*}
\frac{\partial V_*}{\partial \theta_i}(\theta^*)-\lambda\frac{\partial c}{\partial \theta_i}(\theta^*)=-\mu_i^-\leq0.
\end{align*}
Hence
\begin{align*}
\frac{\partial V_*}{\partial \theta_i}(\theta^*)\leq\lambda\frac{\partial c}{\partial \theta_i}(\theta^*).
\end{align*}
Combining the three cases proves the result.
\end{proof}

\section{Poverty Targeting}
\label{sec:app poverty}

\subsection{Background} Proxy means tests are one of the most popular approaches for targeting cash transfer programs and other social protection programs in low-income contexts where comprehensive administrative data on incomes or consumption expenditures are unavailable. They are currently used to target social protection programs in at least fifty low- and middle-income countries \citep{barrientos2018social}. The accuracy of proxy means tests in practice for identifying poor households have been studied in many studies across Africa \citep{brown2018poor, schnitzer2024targeting, mcbride2018retooling}, Asia \citep{narayan2005proxy}, and Latin America \citep{noriega2020algorithmic}. Other work has assessed the fairness of PMTs relative to other targeting approaches \citep{noriega2020algorithmic}, susceptibility to intentional data misreporting and collusion \citep{banerjee2020lack}, and acceptability to data subjects \citep{premand2021efficiency, alatas2012targeting}, but little  previous work has addressed the benefit of investing in PMT coverage and accuracy in comparison relative to investment in allocating resources to the cash transfer programs and other social programs targeted with PMTs. \cite{aiken2025scalable} provide an initial framework for such cost-benefit trade-offs using data on targeting cash transfers in Bangladesh, we expand this framework in this paper to account for additional policy levers and data availability constraints. 

\subsection{Data}
The data for our simulations of cash transfer program are from Ethiopia's 2015 Living Standards Measurement Survey. The survey data covers 4,954 nationally representative households. We use data from 4,694 households that have complete information on consumption expenditures, asset possession, housing quality, and household demographics. We convert consumption expenditures data to yearly total consumption expenditures measured in USD PPP using the PPP exchange rate for 2015 from the World Bank.

\subsection{Predictive Modeling}
The target of our predictive model is log-transformed total household consumption expenditures. The features are 54 variables commonly included in proxy means tests in low-income countries, including:
\begin{itemize}

\item Asset possession, such as as whether the household owns radio, TV, satellite dish, sofa, bicycle, car, or bed
\item Housing quality, such as how many rooms the dwelling has, the construction material of the roof and floor, and the primary materials used for lighting and cooking
\item Household demographics, such as the number of members, number of children, gender and age of the household head, and education and literacy levels of the household head
\item The region of residence of the household

\end{itemize}

All continuous feature variables are winsorized at the 99th percentile and normalized to a 0-1 range. All categorical variables are one hot encoded. 

We train the predictive model on a random 75\% of households and produce predictions on the remaining 25\%. We repeat all our experiments for poverty targeting over 100 random train-test splits and report the average result. Our predictive model is a LASSO regression with regularization strength determined via three-fold cross validation. 

\begin{figure}[t]
\centering
\includegraphics[width=0.45\textwidth]{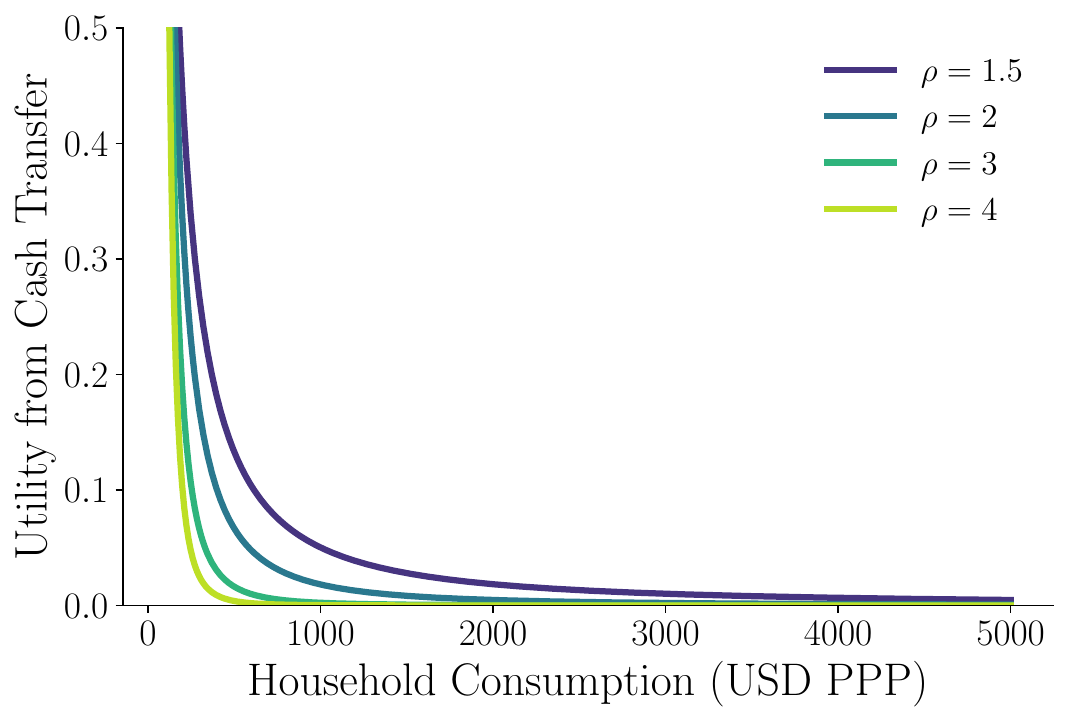}
\caption{Normalized utility gain from a fixed $\$100$ PPP transfer under CRRA utility for different values of $\rho$. Larger values of $\rho$ place relatively more weight on transfers to lower-consumption households.}
\label{fig:crra-utility-app}
\end{figure}

\subsection{Simulating test-time data availability}
\label{app: test-time}
Several of our experiments involve simulating lower test-time data availability -- scenarios where for some households in the test set feature data is not available. We simulate this by removing predictions for these households from the test set; households without predictions receive the population mean as estimate and are targeted at random. 

We also present results for an alternative implementation in which households without scorecard data are treated as ineligible for the program and placed at the end of the allocation queue.

\begin{figure}[t]
\centering
\subfigure[Local Improvements \ref{q1} ]{\includegraphics[width=0.32\textwidth]{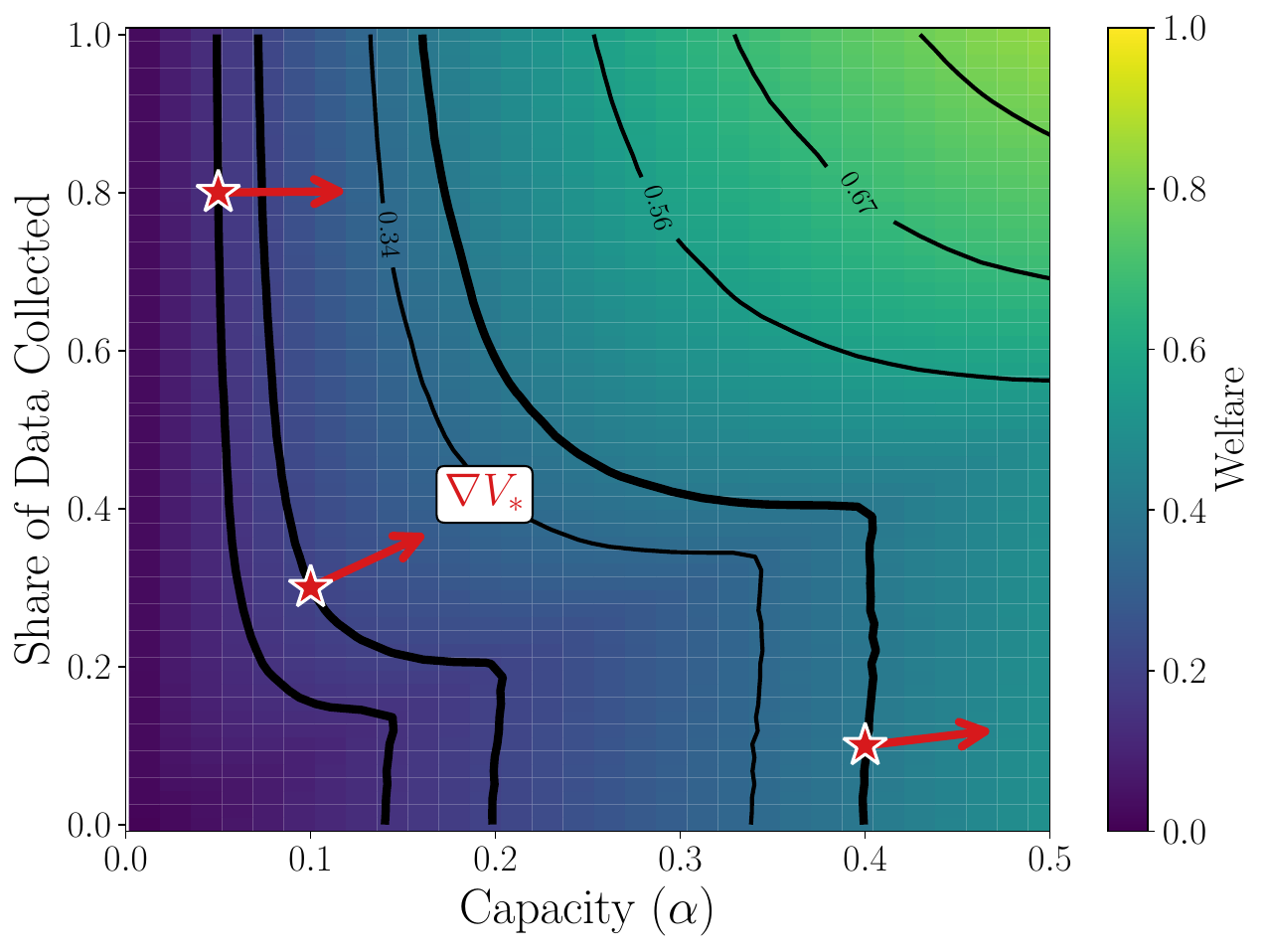}}
\subfigure[Local Optimality \ref{q2}]
{\includegraphics[width=0.32\textwidth]{figures/poverty_step_q2.pdf}}
\subfigure[Expansion Path \ref{q3}]
{\includegraphics[width=0.32\textwidth]{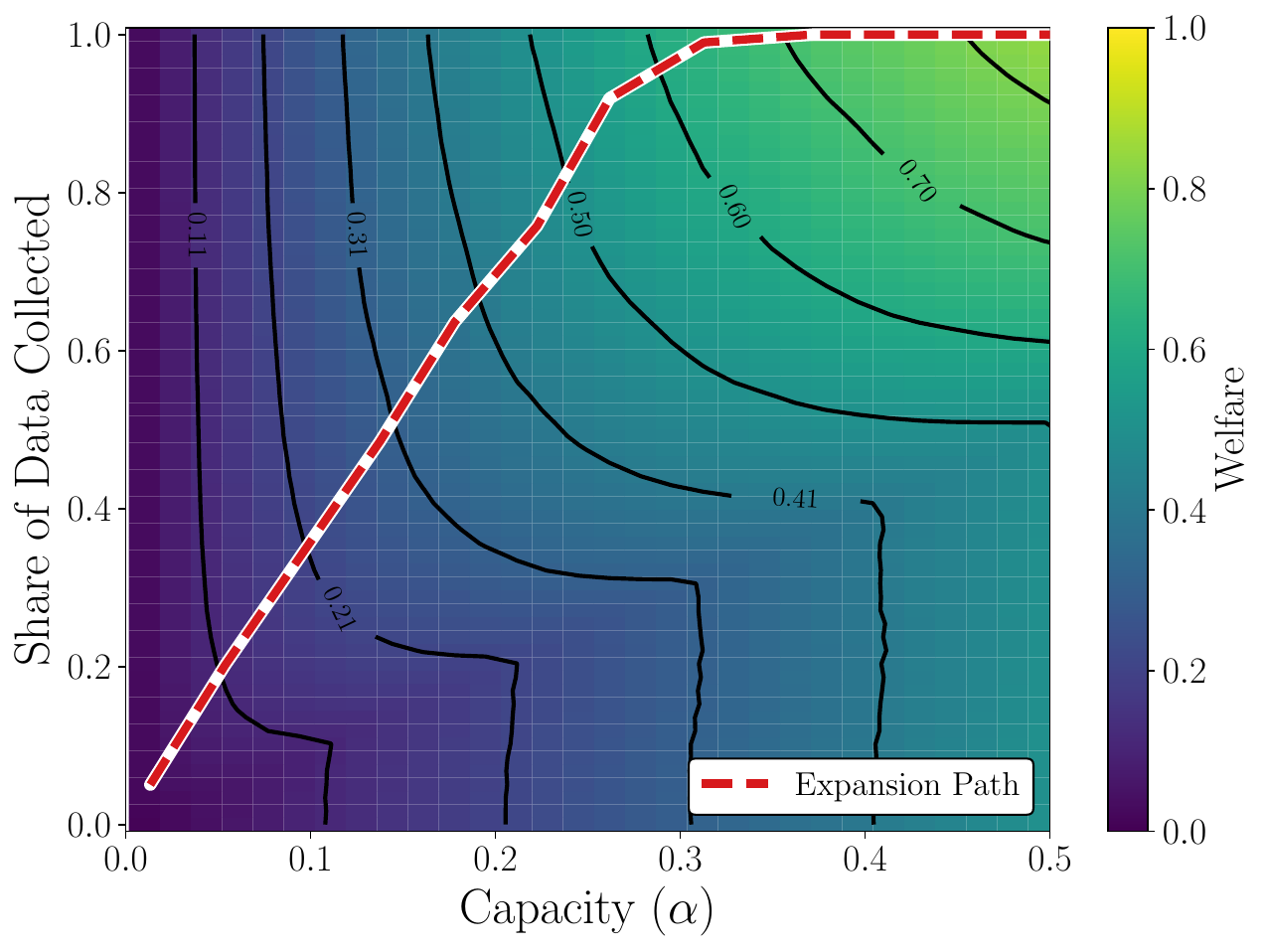}}
\caption{Same as Figure~\ref{fig:step}, but households without scorecard data are placed at the end of the allocation queue rather than assigned the population mean. See Section~\ref{app: test-time} for details.}
\end{figure}

\begin{figure}[t]
\centering
\subfigure[Local Improvements \ref{q1}]{\includegraphics[width=0.32\textwidth]{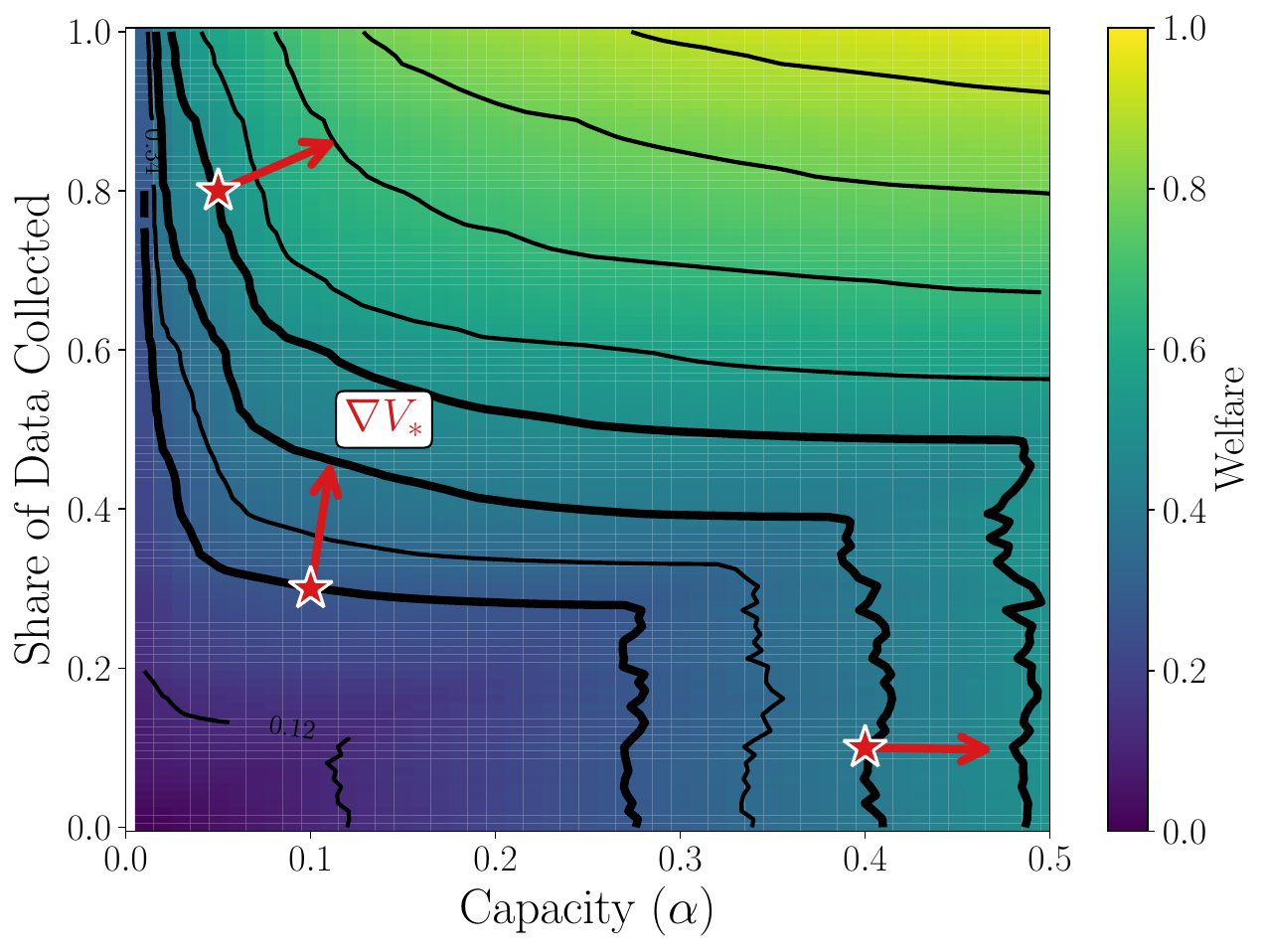}}
\subfigure[Local Optimality \ref{q2}]
{\includegraphics[width=0.32\textwidth]{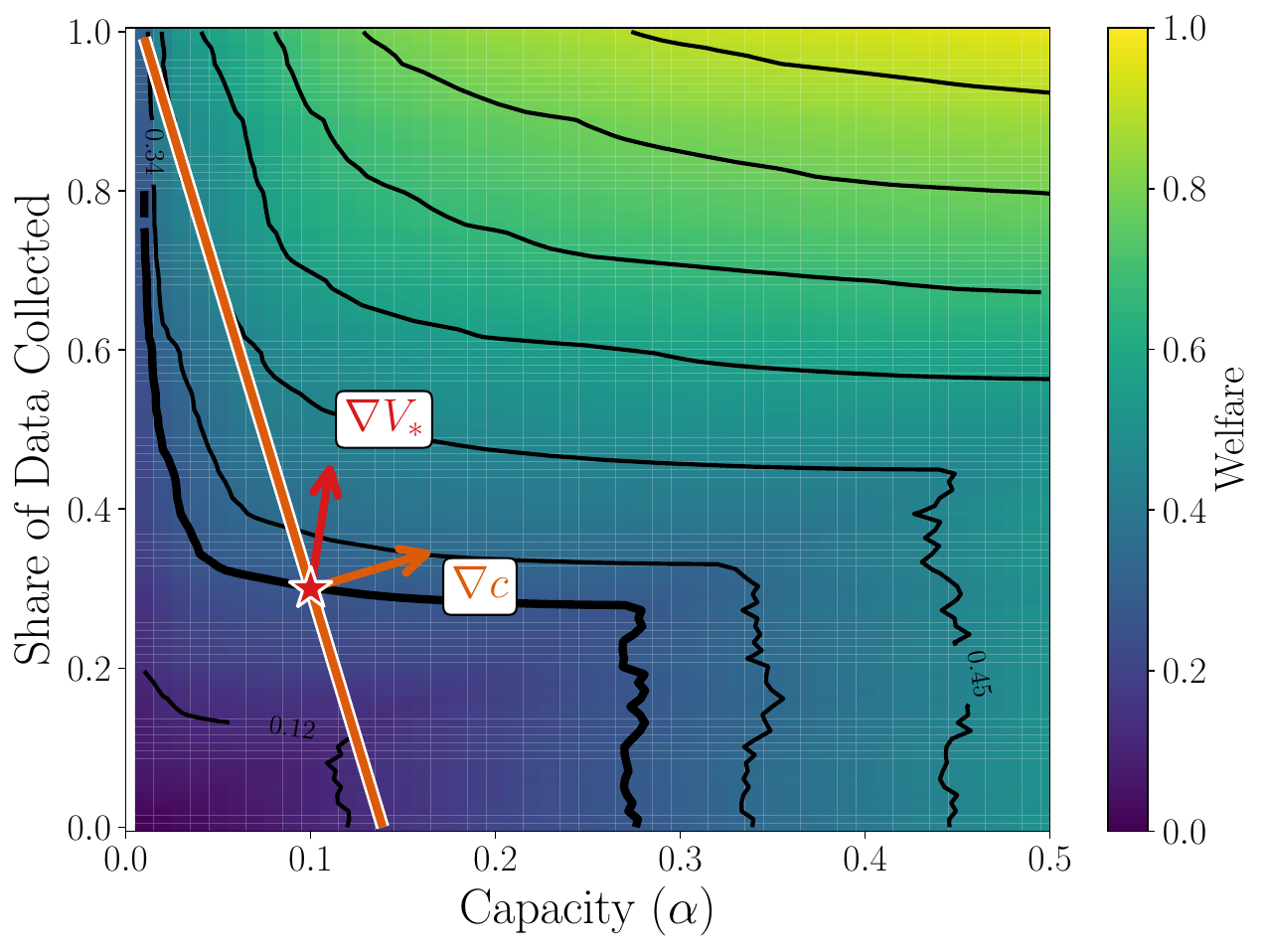}}
\subfigure[Expansion Path \ref{q3} ]
{\includegraphics[width=0.32\textwidth]{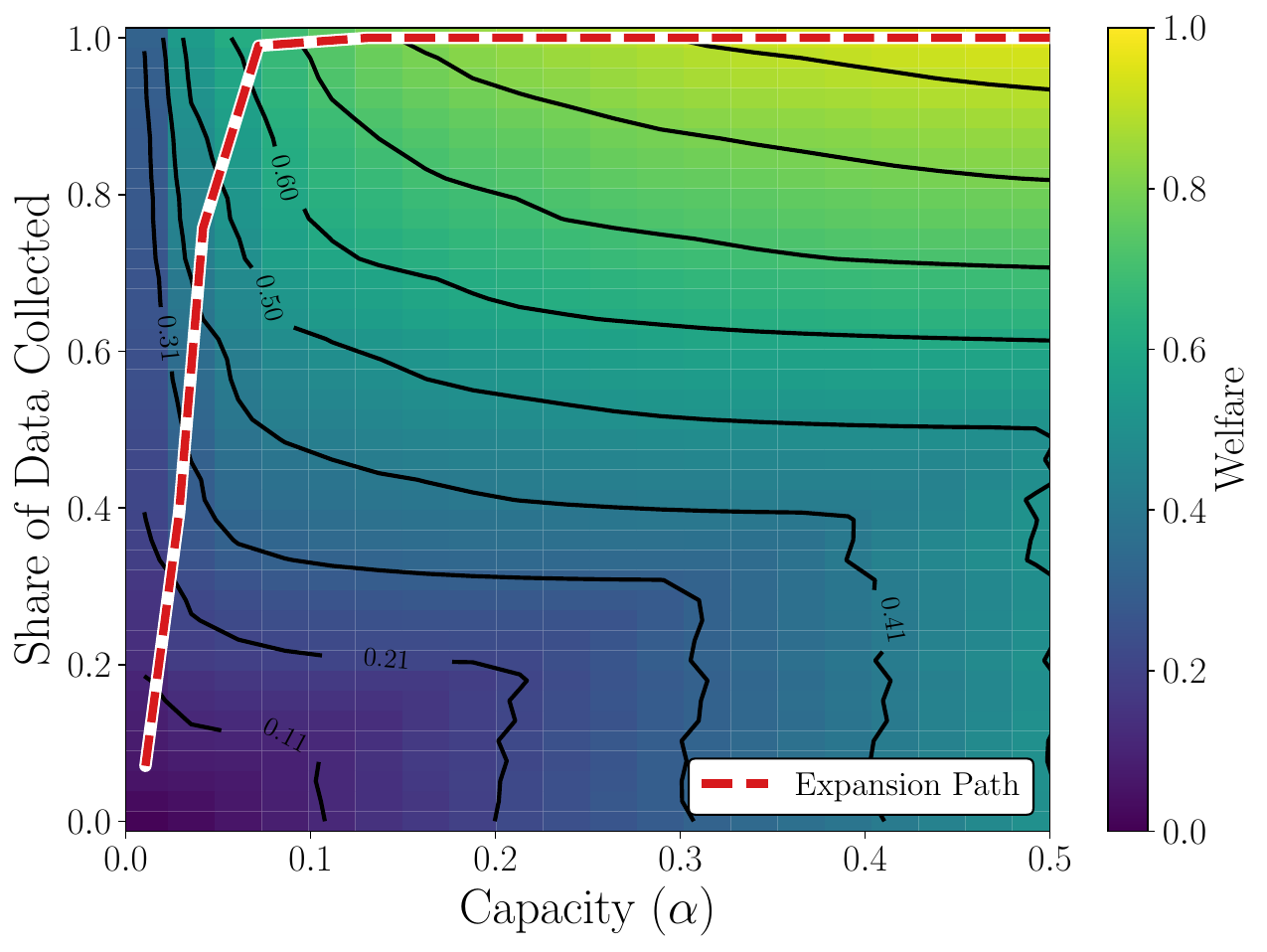}}
\caption{Same as Figure~\ref{fig:crra}, but households without scorecard data are placed at the end of the allocation queue rather than assigned the population mean. See Section~\ref{app: test-time} for details.}
\end{figure}

\begin{figure}[t]
\centering
\subfigure[Local Improvements \ref{q1}]{\includegraphics[width=0.32\textwidth]{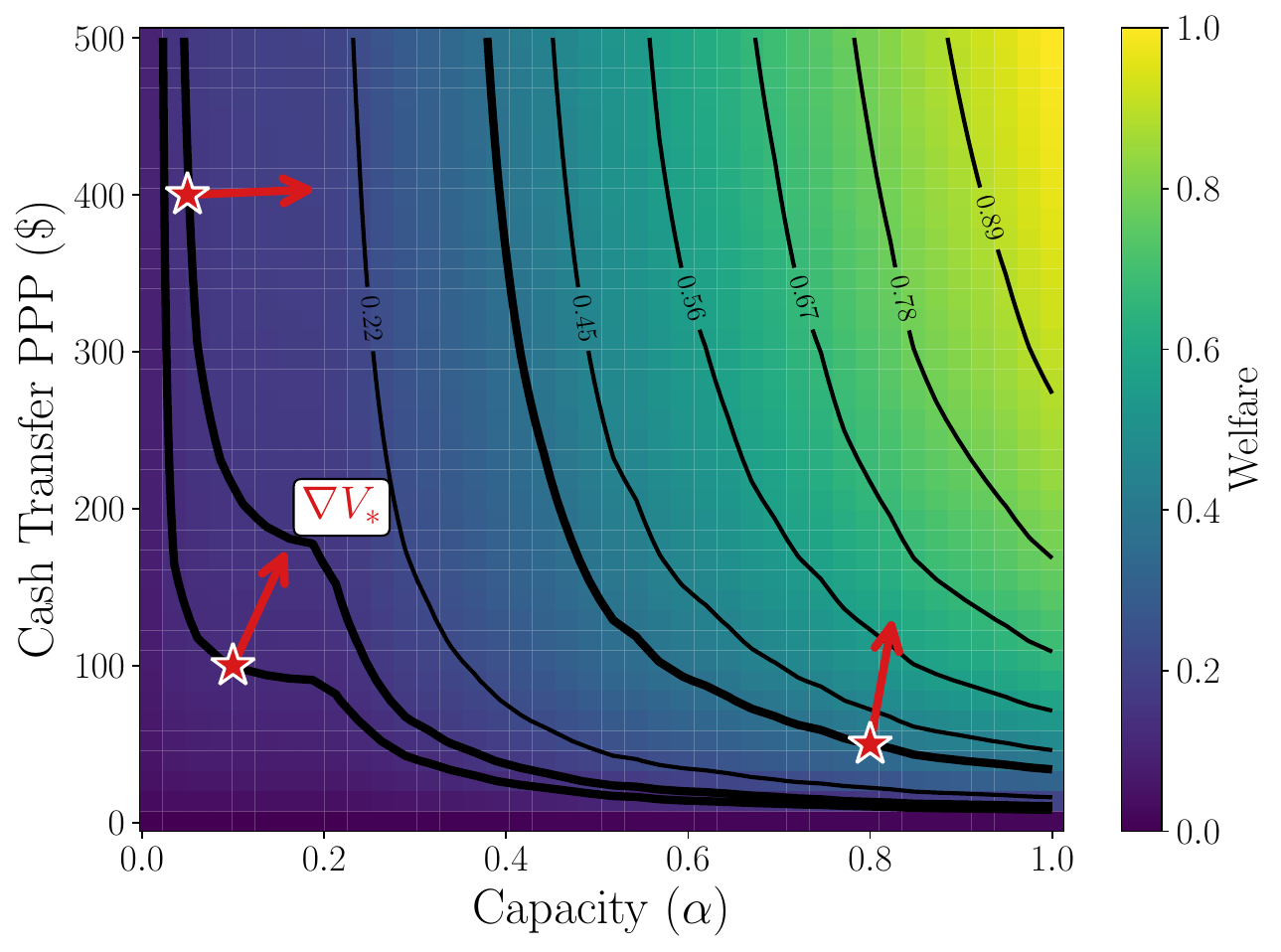}}
\subfigure[Local Optimality \ref{q2}]
{\includegraphics[width=0.32\textwidth]{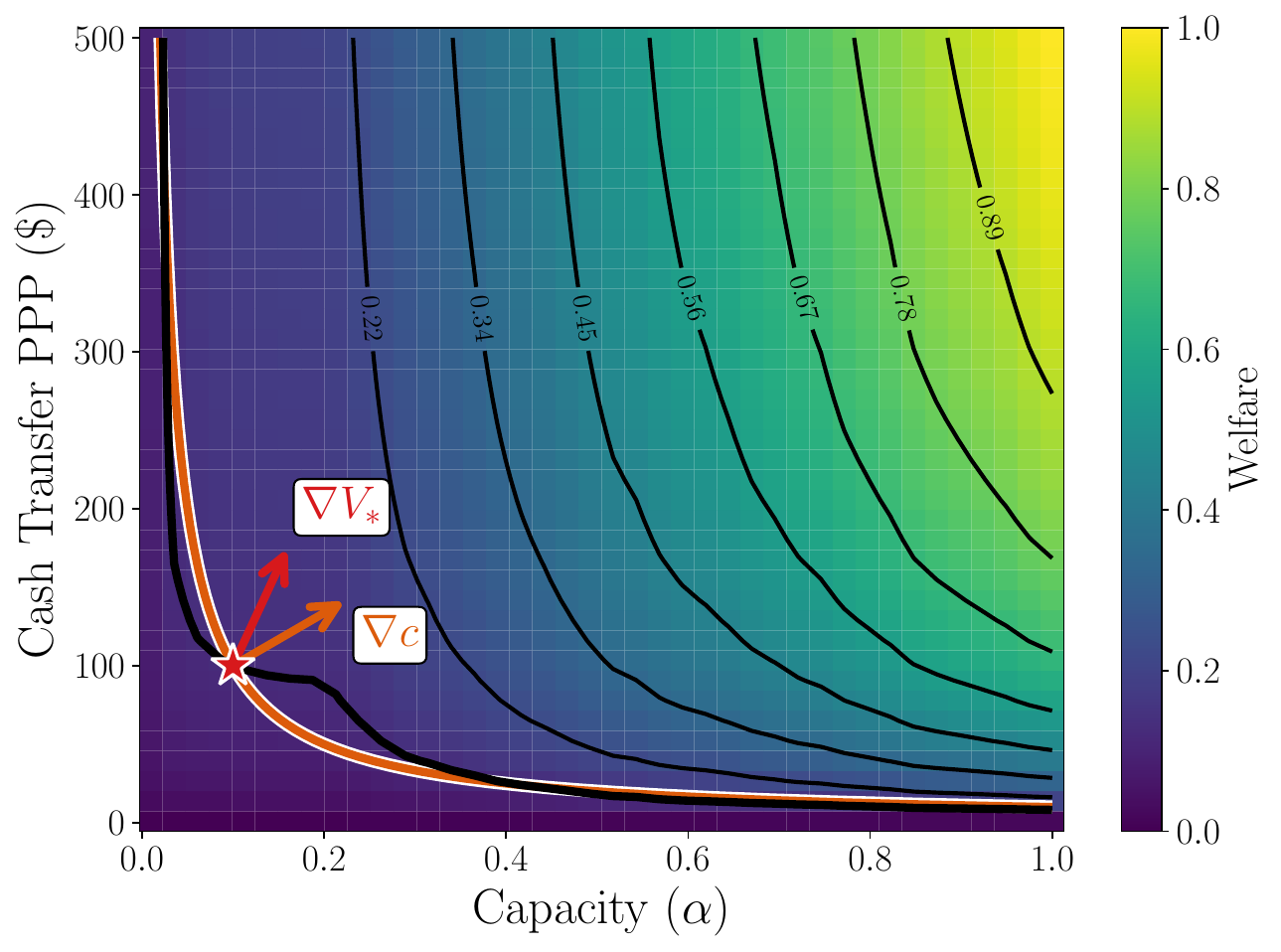}}
\subfigure[Expansion Path \ref{q3} ]
{\includegraphics[width=0.32\textwidth]{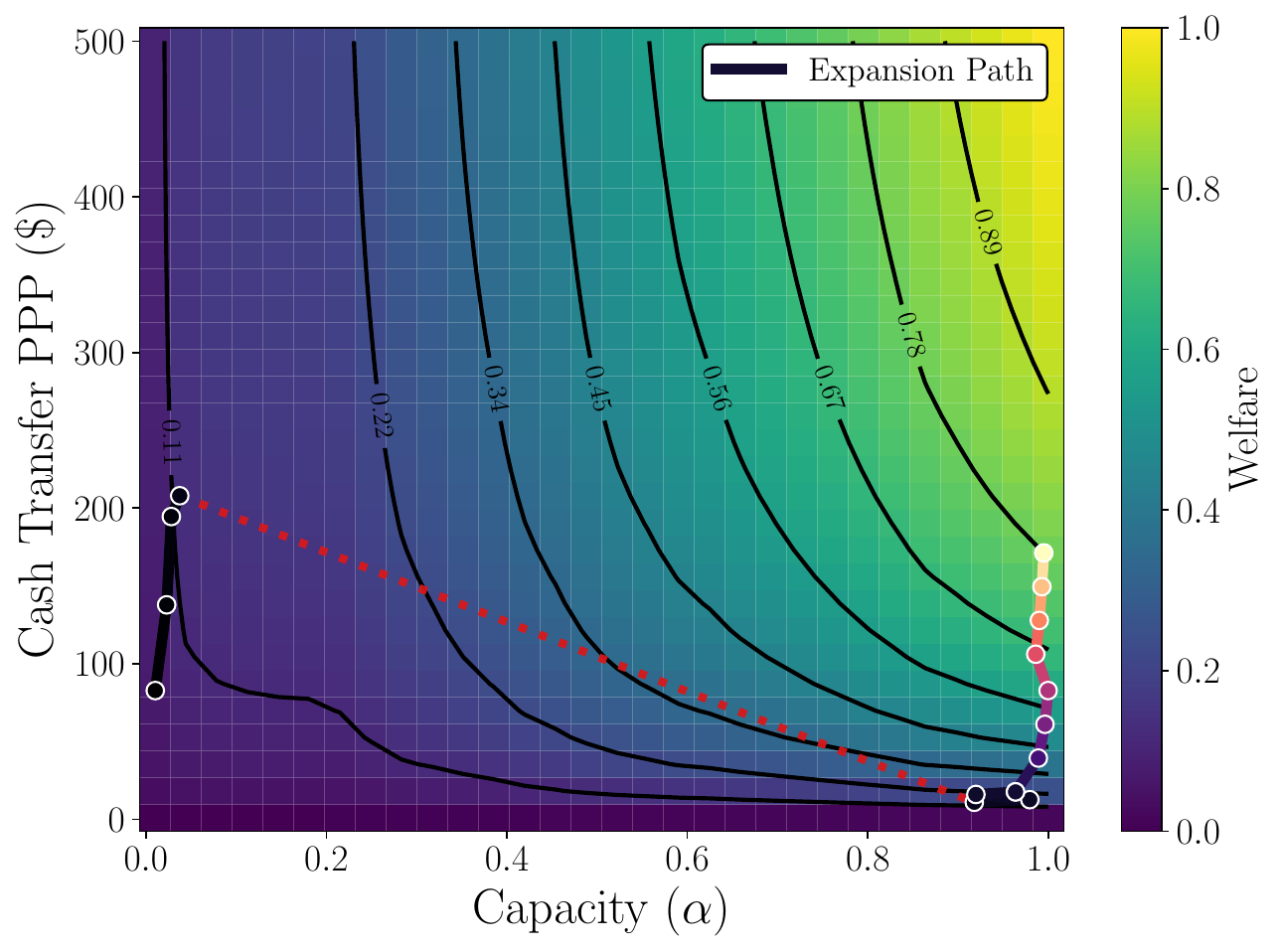}}
\caption{Same as Figure~\ref{fig:benefit}, but households without scorecard data are placed at the end of the allocation queue rather than assigned the population mean. See Section~\ref{app: test-time} for details.}
\end{figure}

\subsection{Sensitivity to cost and utility parameters}

\begin{figure}[H]
\centering
\subfigure[$\rho=2$]{\includegraphics[width=0.32\textwidth]{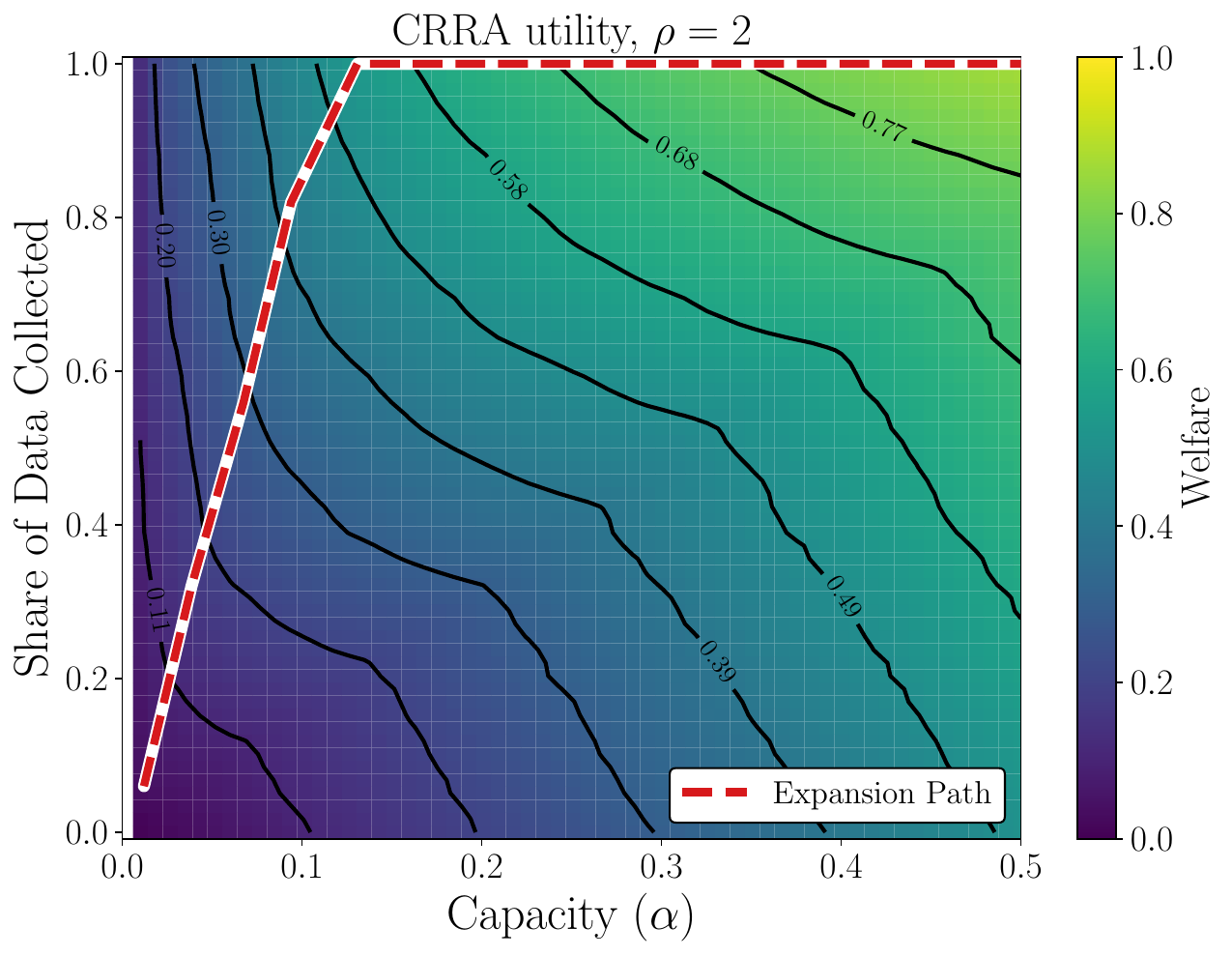}}
\subfigure[$\rho=3$ (baseline)]{\includegraphics[width=0.32\textwidth]{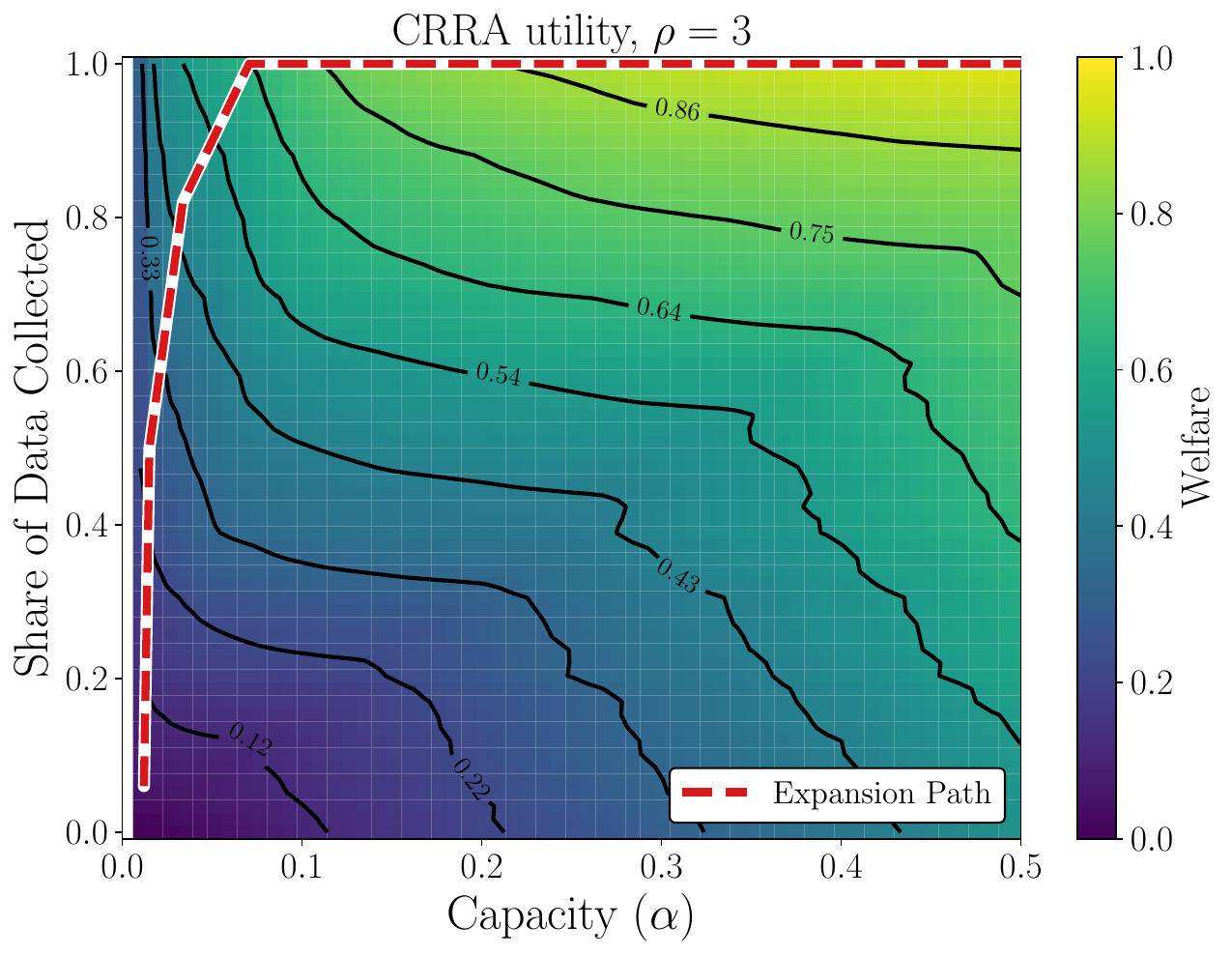}}
\subfigure[$\rho=4$]{\includegraphics[width=0.32\textwidth]{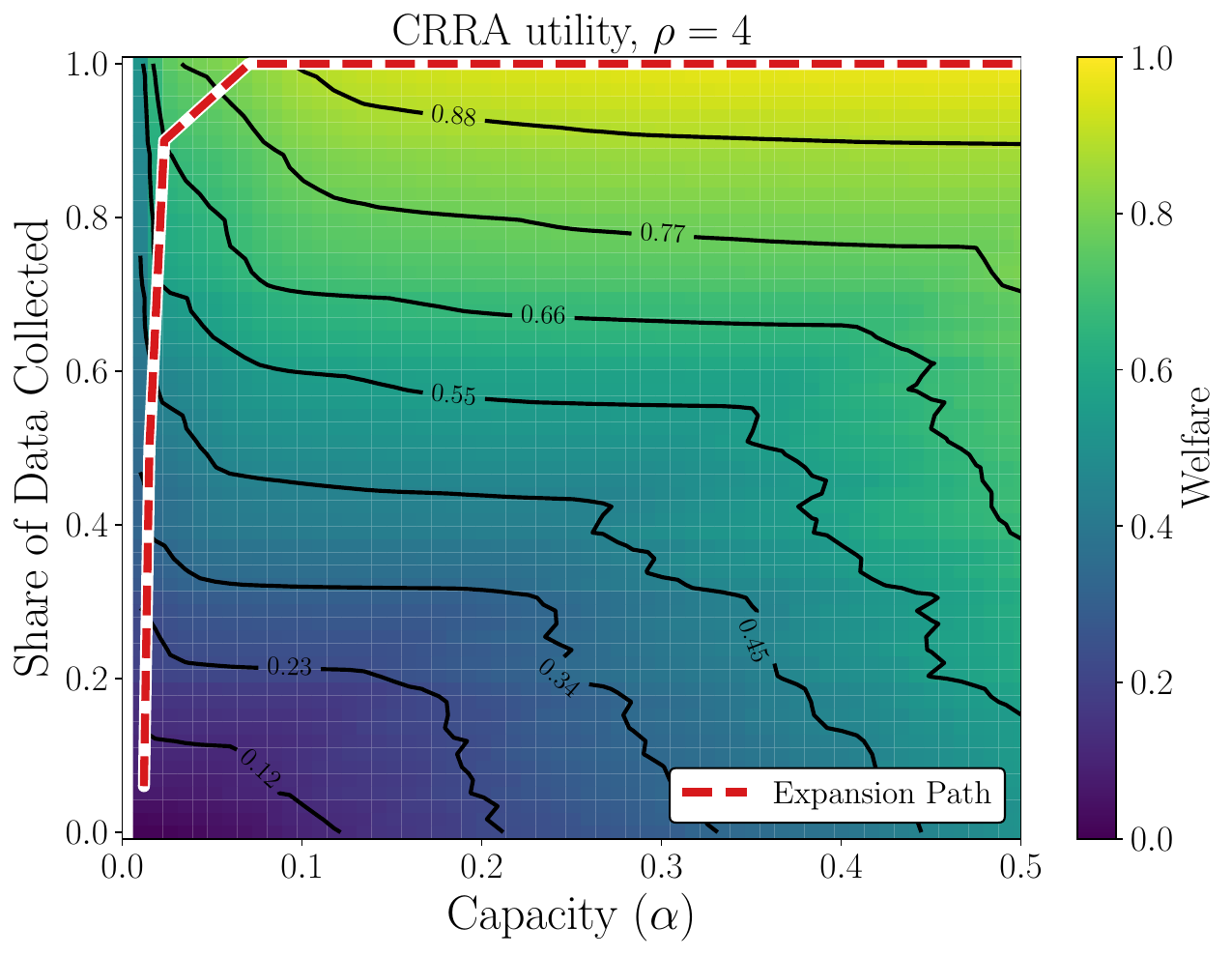}}
\caption{Sensitivity of the expansion path to CRRA curvature. Welfare surfaces and optimal expansion paths (red dashed) for $\rho \in \{2, 3, 4\}$.}
\label{fig:poverty-rho-sensitivity}
\end{figure}

\begin{figure}[H]
\centering
\subfigure[\$6.5)]{\includegraphics[width=0.32\textwidth]{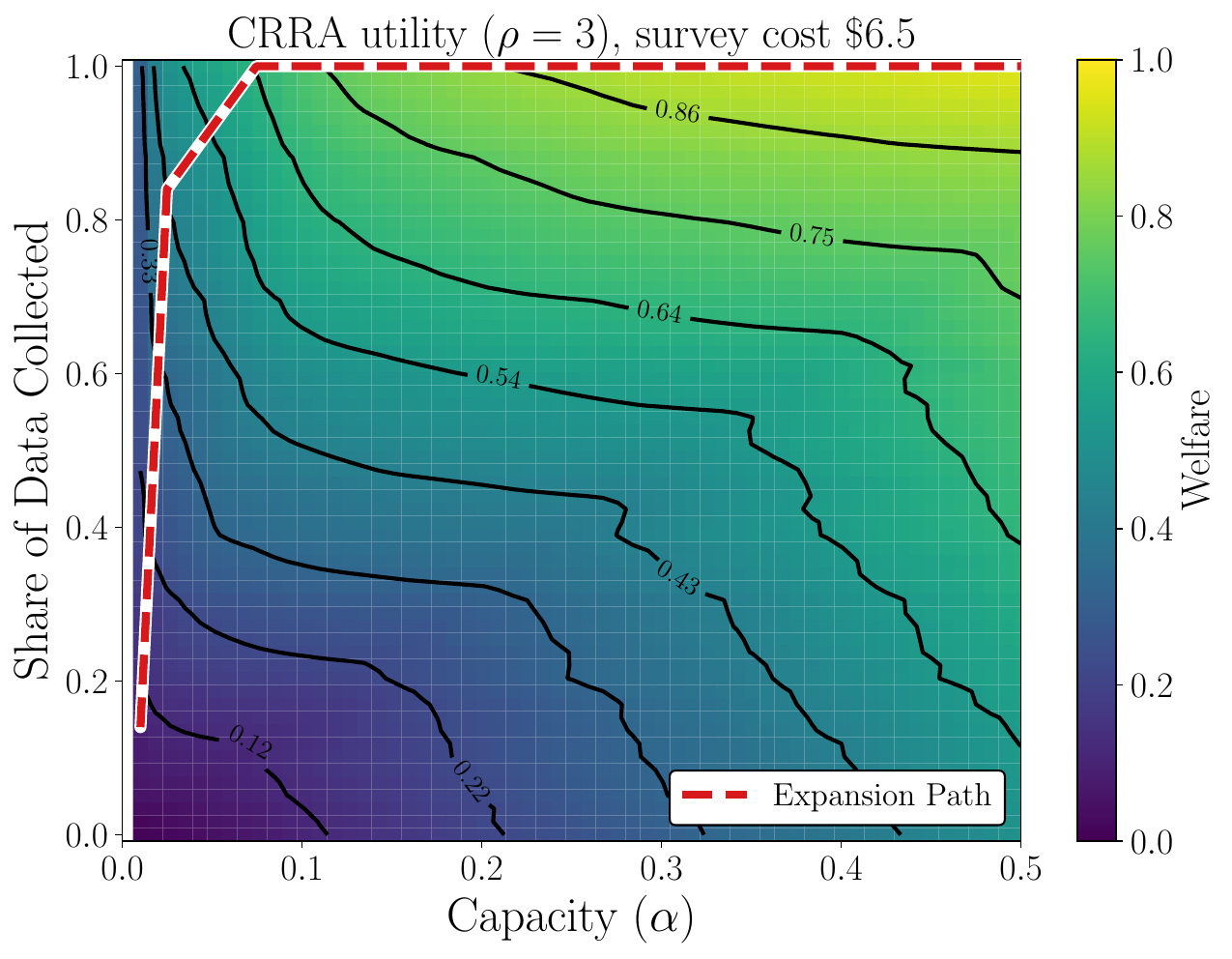}}
\subfigure[\$13)]{\includegraphics[width=0.32\textwidth]{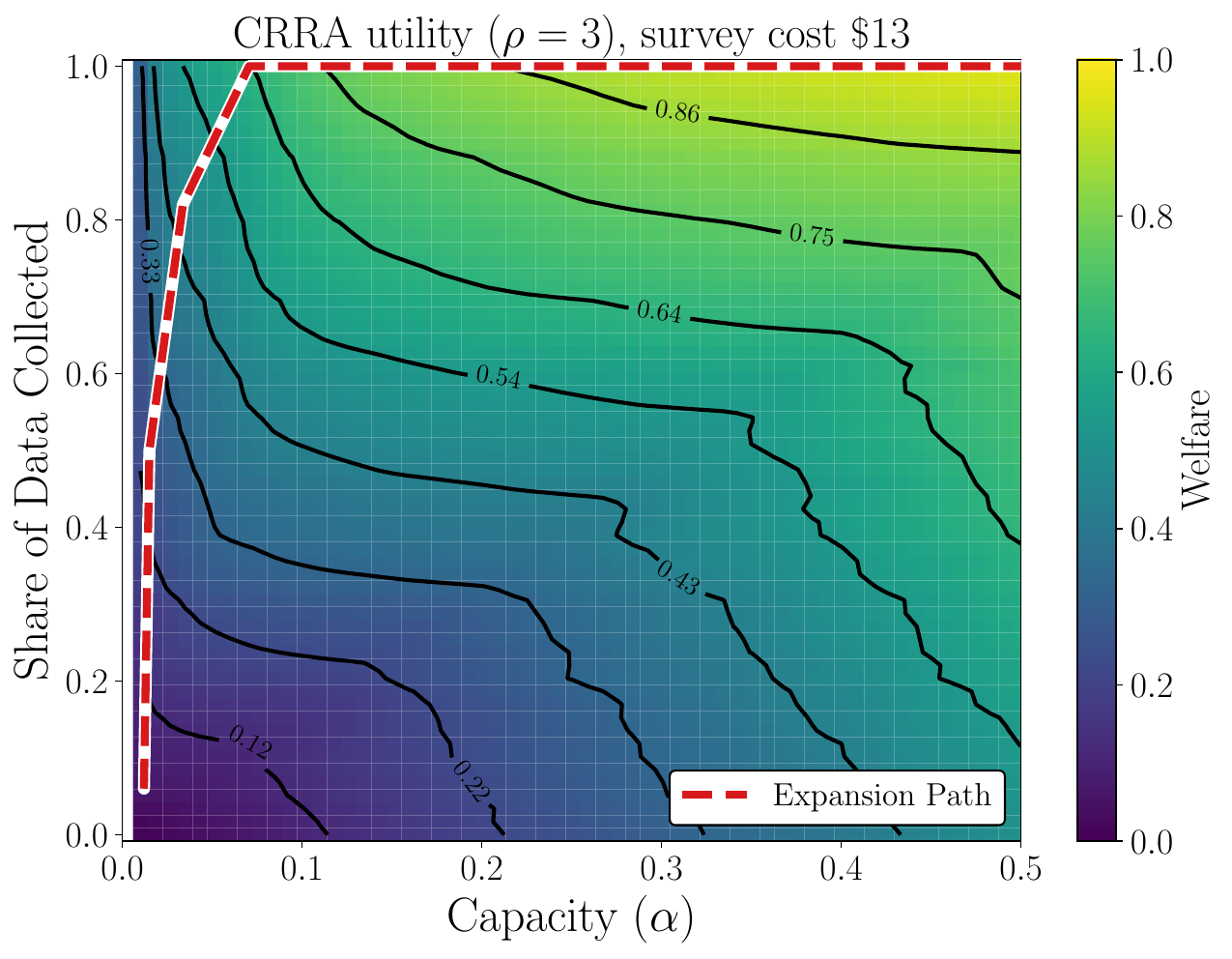}}
\subfigure[\$19.5]{\includegraphics[width=0.32\textwidth]{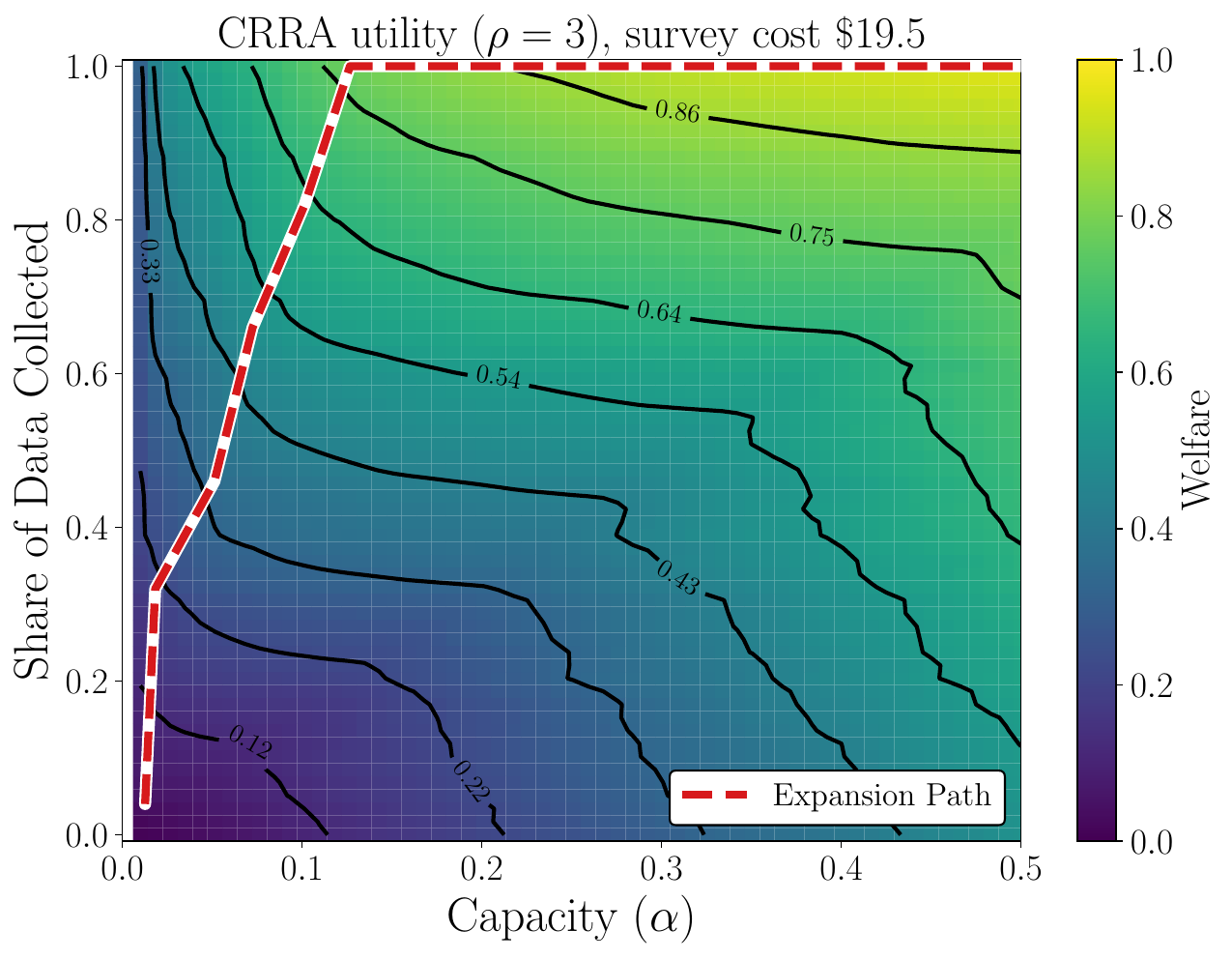}}
\caption{Sensitivity of the expansion path to survey costs. Welfare surfaces and optimal expansion paths (red dashed) under CRRA utility ($\rho=3$) with the per-household survey cost perturbed by $\pm 50\%$ around the \$13 baseline.}
\label{fig:poverty-cost-sensitivity}
\end{figure}

\subsection{Replication}
All code for the poverty targeting simulations is available, including cleaning the survey data to prepare it for predictive modeling. The survey data can be downloaded from the LSMS program at \url{https://microdata.worldbank.org/index.php/catalog/2783}. 

\subsection{Compute Resources}
\label{compute}
All experiments (poverty-targeting and employment office case studies) were conducted on a consumer laptop. No GPU resources were used. Individual experimental runs (corresponding to individual figures in the paper) completed in under 10 minutes each, with most finishing in seconds. The total compute required to reproduce all reported results is under one hour on comparable hardware.

\begin{figure}[t]
\centering
\includegraphics[width=0.62\textwidth]{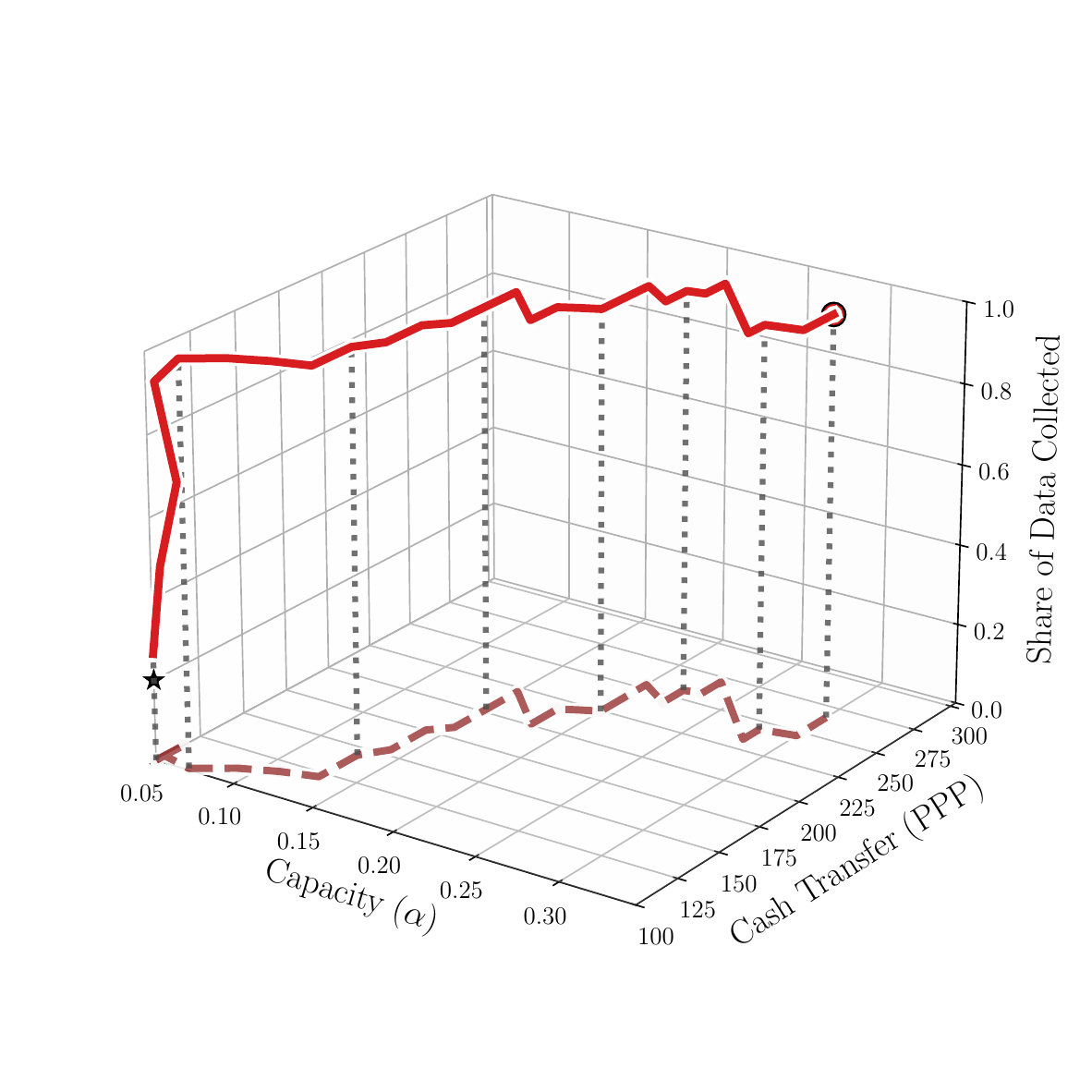}
\caption{Expansion path for optimizing over program capacity, transfer size, and data collection starting at a transfer size of $\$100$ PPP, $5\%$ program capacity and $20\%$ of poverty score cards collected. The red curve traces the welfare-maximizing design as the available budget increases.}
\label{fig:benefit-capacity-label}
\end{figure}

\section{Identifying Long-Term Unemployment in Germany}
\label{sec: app a}

\subsection{Background}

Algorithmic profiling of job seekers in the delivery of support measures are debated and implemented by Public Employment Services (PES) in various countries \citep{loxha2014profiling, kortner2023predictive}. In Germany, the introduction of algorithmic decision-support tools has been discussed since the early 2000s \citep{rudolph2001profiling}. A common narrative in these discussions is the hope to improve efficiency over rule- or case-worker-based profiling workflows, supported by comparative studies that report higher accuracy of algorithmic profiling relative to human case-workers in predicting long-term unemployment risks \citep{desiere2019statistical, junquera2025rules}. Yet, German PES currently rely on case-worker-based profiling, motivating our case study as an illustrative investigation of how different profiling regimes could support employment services in Germany. 

\citet{zezulka2024fair} study the link between LTU risk predictions and targeting decisions with data from the Swiss PES. In the targeting context, \citet{cockx2020priority} and \citet{goller2023} exploit effect heterogeneity of active labour market programs (ALMPs) to find optimal allocations of programs to job seekers. \citet{kortner2023inequality} demonstrate how inequality-averse allocation of ALMPs to job seekers could mitigate historical differences in group-specific unemployment risks. Yet, a common thread in these works is that most parameters of the allocation problem are treated as given, i.e, they focus on modeling solutions rather than asking how outcomes could be improved with an expanded set of actions, such as additional data collection or increased screening capacity.

\subsection{Data}
\label{app: emply data}
We make use of a comprehensive dataset on German jobseekers. The Sample of Integrated Labour Market Biographies (SIAB) is provided via a Scientific Use File from the Research Data Centre (FDZ) of the German Federal Employment Agency (BA) at the Institute for Employment Research (IAB) \citep{SchmuckerVomBerge2023SIABR7521v1}. The SIAB is a 2\% sample of all governmental labor market records in Germany, comprising approximately 60 million entries spanning 1970 to 2021. It collects information on employment status, jobseeker demographics, employment and benefits histories, and participation in active labor market programs, from a range of different data sources across the German employment services (e.g., the Employee History (BeH), Benefit Recipient History (LeH) and Unemployment Benefit II Recipient History (LHG) are separate data products within the BA). Further details on its content, sampling procedure, and anonymization can be found in the accompanying data report \citep{SchmuckerVomBerge2023FDZD2307}.

\begin{table}[htbp]
\centering
\caption{Feature groups}
\label{tab:feature-groups}
\begin{tabularx}{\textwidth}{@{} l X @{}}
\toprule
\textbf{Group} & \textbf{Feature} \\
\midrule
\textit{Socio-demographics \& residence} & Age \\
                 & Gender \\
                 & Educational background \\
                 & State of residence \\
                 & German nationality \\
                 & Commuting status and residential moves  \\
\midrule
\textit{Pre-entry status} & Employment in six weeks before unemployment entry \\
                          & Benefit receipt in six weeks before entry \\
                          & Registered job-seeking in six weeks before entry \\
                          & Registered with public employment office for other reasons in six weeks before entry \\
                          & Subsidized employment in six weeks before entry \\
                          & Missing information for the six weeks before entry \\       
\midrule
\textit{Last job} & Availability of previous-job information \\
                  & Duration of last job \\
                  & Parallel employment during last job \\
                  & Type of last job \\
                  & Part-time status of last job \\
                  & Daily wage in last job \\
                  & Occupational skill level of last job \\
                  & Fixed-term and temporary-agency status of last job \\
                  & Industry of last job \\

\midrule
\textit{Employment history} & Number of previous jobs \\
                            & Number of previous establishments \\
                            & Total and mean duration of previous employment \\
                            & Duration in marginal, full-time, parallel, fixed-term, and temporary-agency employment \\
                            & Duration worked across industries \\
                            & Time since first employment \\
                            & Time since last contact with labor market\\
\midrule
\textit{Agency history} & Previous unemployment and job-search history \\
                        & Benefit receipt history \\
                     & Participation in active labor market programs \\
                     & Subsidized employment history \\
                     & Time since last job search \\
\bottomrule
\end{tabularx}
\end{table}

\subsection{Predictive Modeling} We construct a prediction task for long-term unemployment, predicting the number of months a job seeker will remain unemployed at entry into unemployment. Following \citet{bach2023impact, kern2024, pmlr-v267-fischer-abaigar25a}, we use the same spell aggregation procedure to construct covariates capturing sociodemographic information, labor market and benefits history, and characteristics of the last job held. This leads to 56 numerical and 24 categorical variables; the latter are one-hot encoded. To simulate deployment conditions in real employment offices, we use records from 2010 and 2011 for training, data from 2012 for validation, and data from 2015 for testing. We train a gradient boosting model (CatBoost) on this temporal split. All policy lever comparisons in this paper are conducted on the test set ($n=86,692$).

\subsection{Replication} We provide full replication code for the design space analysis. Due to data sensitivity, the underlying records are not publicly available but can be applied for at the Research Data Centre (FDZ) of the Institute for Employment Research (IAB). 

\section{Societal Impact}
\label{impact}
This work addresses the design of prediction-based allocation systems in public institutions, where algorithmic choices directly affect access to social goods. Our framework helps practitioners understand when prediction improvements are worthwhile relative to other investments, enabling context-specific comparisons that can inform resource allocation decisions. The framework adopts the perspective of a social planner optimizing a specified objective. This carries limitations. First, the analysis requires specifying utility functions and cost structures, excluding considerations that resist quantification, such as procedural fairness. Second, institutional objectives may diverge from the welfare of affected populations; our toolkit can inform institutional decision-making but does not resolve this tension.

\end{document}